\documentclass{article}
\usepackage[a4paper]{geometry}
\usepackage{stmaryrd}
\usepackage{amsmath}
\usepackage{amssymb}
\usepackage{mathabx}
\usepackage{graphicx}
\usepackage{tikz}
\usepackage{colonequals}
\usepackage{microtype}
\usepackage[noconfig]{ntheorem}
\usepackage{xspace}
\usepackage{verbatim} 
\usepackage{subfigure}
\usepackage{url}
\usepackage{float}
\newtheorem{theorem}{Theorem}[section]
\newtheorem{lemma}[theorem]{Lemma}
\newtheorem{proposition}[theorem]{Proposition}
\newtheorem{corollary}[theorem]{Corollary}
\newtheorem{definition}[theorem]{Definition}
\newtheorem{example}[theorem]{Example}
\newtheorem*{proof}{Proof}

\newcommand{\Paths}{\ensuremath{\mathit{Paths}}}
\newcommand{\LabPaths}{\ensuremath{\mathit{LabPaths}}}
\newcommand{\shuffle}{\mathbin{\|}}
\renewcommand{\shuffle}{\mathbin{|\hspace{-0.1em}|}}
\newcommand{\dblslash}{\ensuremath{/\!/}}
\newcommand{\wc}{\mathord{\star}}

\newcommand{\lab}{\mathit{lab}}
\renewcommand{\root}{\mathit{root}}

\newcommand{\child}{\mathit{child}}
\newcommand{\ch}{\mathit{ch}}

\newcommand{\desc}{\mathit{desc}}
\newcommand{\universal}{\mathit{non\_nullable}}
\newcommand{\existential}{\mathit{nullable}}
\newcommand{\minnb}{\mathit{min\_nb}}

\newcommand{\Twig}{\ensuremath{\mathit{Twig}}}
\newcommand{\Tree}{\ensuremath{\mathit{Tree}}}
\newcommand{\CNF}{\mathit{3CNF}}
\newcommand{\sat}{\mathrm{3SAT}}
\newcommand{\CNT}{\mathrm{CNT}}
\newcommand{\SAT}{\mathrm{SAT}}
\newcommand{\IMPL}{\mathrm{IMPL}}

\newcommand{\MEMB}{\mathrm{MEMB}}
\newcommand{\DTD}{\mathit{DTD}}
\newcommand{\ms}{\mathit{MS}}
\newcommand{\dms}{\mathit{DMS}}

\newcommand{\cardinality}{\mathit{cardinality}}
\newcommand{\conflict}{\mathit{conflict}}
\newcommand{\required}{\mathit{required}}
\newcommand{\mcount}{\mathit{count}}
\newcommand{\pconflict}{\mathit{present\_conflict}}
\newcommand{\prequired}{\mathit{present\_required}}

\newcommand{\open}{\mathit{open\_tag}}
\newcommand{\close}{\mathit{close\_tag}}

\newenvironment{changemargin}[2]{%
  \begin{list}{}{%
    \setlength{\topsep}{0pt}%
    \setlength{\leftmargin}{#1}%
    \setlength{\rightmargin}{#2}%
    \setlength{\listparindent}{\parindent}%
    \setlength{\itemindent}{\parindent}%
    \setlength{\parsep}{\parskip}%
  }%
  \item[]}{\end{list}}

\def\qed {{                
   \parfillskip=0pt        
   \widowpenalty=10000     
   \displaywidowpenalty=10000  
   \finalhyphendemerits=0  
   \leavevmode             
   \unskip                 
   \nobreak                
   \hfil                   
   \penalty50              
   \hskip.2em              
   \null                   
   \hfill                  
   $\square$
   \par}}                  

\floatstyle{ruled}
\newfloat{float@algorithm}{htbp}{loa}
\floatname{float@algorithm}{Algorithm}
\newcounter{LineCounter@algorithm} 

\newenvironment{BasicCommands@algorithm}{%
  \newcommand{\TAB}{\makebox[4ex][r]{}}%
  \newcommand{\ALGORITHM}{\textbf{algorithm}\xspace}%
  \newcommand{\INPUT}{\textbf{Input}\xspace}%
  \newcommand{\OUTPUT}{\textbf{Output}\xspace}%
  \newcommand{\IF}{\textbf{if}\xspace}%
  \newcommand{\THEN}{\textbf{then}\xspace}%
  \newcommand{\AND}{\textbf{and}\xspace}%
  \newcommand{\REJECT}{\textbf{reject}\xspace}%
  \newcommand{\ACCEPT}{\textbf{accept}\xspace}%
}{%
}

\newenvironment{algorithm*}%
{%
\begin{BasicCommands@algorithm}%
\list{}{\itemindent 0em%
        \listparindent\itemindent
        \rightmargin  \leftmargin}%
\item\relax
}{%
\endlist
\end{BasicCommands@algorithm}%
}

\newenvironment{algorithm}%
{%
  \newcommand{\ResetLineCounter}{%
    \setcounter{LineCounter@algorithm}{0}%
  }%
  \newcommand{\LN}{%
    \makebox[0pt][l]{%
      \makebox[0pt][l]{%
        \addtocounter{LineCounter@algorithm}{1}%
      }%
      \makebox[10.75pt][r]{
        \fontsize{7}{10}%
        \selectfont%
        \arabic{LineCounter@algorithm}:%
      }%
    }%
    \TAB%
  }%
  \begin{BasicCommands@algorithm}%
  \begin{float@algorithm}%
    \ResetLineCounter%
}%
{%
  \end{float@algorithm}%
  \end{BasicCommands@algorithm}%
}
\begin{document}
\title{Simple Schemas for Unordered XML}
\author{Iovka Boneva \and Radu Ciucanu \and S\l awek Staworko}
\date{University of Lille \& INRIA, France}
\maketitle
\thispagestyle{empty}
\begin{abstract}
  We consider unordered XML, where the relative order among siblings
  is ignored, and propose two simple yet practical schema formalisms: 
  \emph{disjunctive multiplicity schemas} (DMS),
  and its restriction, \emph{disjunction-free multiplicity schemas}
  (MS). We investigate their computational properties and characterize
  the complexity of the following static analysis problems: schema
  satisfiability, membership of a tree to the language of a schema,
  schema containment, twig query satisfiability, implication, and
  containment in the presence of schema. Our research indicates that
  the proposed formalisms retain much of the expressiveness of DTDs
  without an increase in computational complexity.
\end{abstract}

\section{Introduction}
When XML is used for \emph{document-centric} applications, the
relative order among the elements is typically important e.g., the
relative order of paragraphs and chapters in a book. On the other
hand, in case of \emph{data-centric} XML applications, the order among
the elements may be unimportant~\cite{AbBoVi12}. In this paper we
focus on the latter use case. As an example, take a trivialized
fragment of an XML document containing the DBLP repository in
Figure~\ref{fig:dblp}. While the order of the elements {\sl title},
{\sl author}, and {\sl year} may differ from one publication to
another, it has no impact on the semantics of the data stored in this
semi-structured database.

A \emph{schema} for XML is a description of the type of admissible
documents, typically defining for every node its \emph{content model}
i.e., the children nodes it must, may, or cannot contain. For instance,
in the DBLP example, we shall require every \textsl{article} to have
exactly one \textsl{title}, one \textsl{year}, and one or more
\textsl{author}'s. A \textsl{book} may additionally contain one
\textsl{publisher} and may also have one or more \textsl{editor}'s
instead of \textsl{author}'s. A schema has numerous important uses.
For instance, it allows to validate a document against a schema and
identify potential errors. A schema also serves as a reference for any
user who does not know yet the structure of the XML document and
attempts to query or modify its contents.

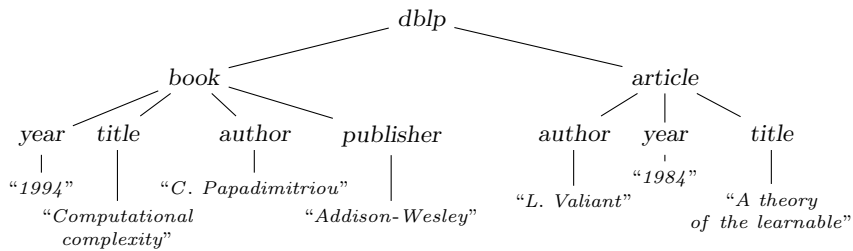
\begin{figure}
  \centering
  \begin{tikzpicture}[yscale=0.75,xscale=2]\small
		\node at (0,0) (n0) {\sl dblp};
		\node at (-1.5,-1) (n10) {\sl book}edge[-] (n0);
		\node at (-2.5, -2.1) (m111) {\sl year}edge[-] (n10);
		\node at (-2, -2) (m112) {\sl title}edge[-] (n10);
		\node at (-1.1, -2) (m113) {\sl author}edge[-] (n10);
		\node at (-0.2, -2.05) (m114) {\sl publisher}edge[-] (n10);
		\node at (-2.5,-3) {\scriptsize ``$\mathit{1994}$''} edge[-](m111);
		\node at (-2,-3.5) {\scriptsize ``$\mathit{Computational}$} edge[-] (m112);
		\node at (-2,-3.9) {\scriptsize $\mathit{complexity}$''};
		\node at (-1.1,-3) {\scriptsize ``$\mathit{C.\  Papadimitriou}$''} edge[-] (m113);
		\node at(-0.2, -3.5) {\scriptsize ``$\mathit{Addison}$-$\mathit{Wesley}$''} edge[-] (m114);
		\node at (1.6,-1) (n11) {\sl article}edge[-] (n0);
		\node at (1, -2) (m111) {\sl author}edge[-] (n11);
		\node at (1.6, -2.1) (m112) {\sl year}edge[-] (n11);
		\node at (2.3, -2) (m113) {\sl title}edge[-] (n11);
		\node at (1,-3.2) {\scriptsize``$\mathit{L.\  Valiant}$''} edge[-](m111);
		\node at (1.6,-2.8) {\scriptsize``$\mathit{1984}$''} edge[-] (m112);
		\node at (2.3,-3.2) {\scriptsize``$\mathit{A\  theory}$} edge[-] (m113);
		\node at (2.3,-3.6) {\scriptsize $\mathit{of\ the\ learnable}$''};
  \end{tikzpicture}
  \caption{\label{fig:dblp}A trivialized DBLP repository.}
\end{figure}

The \emph{Document Type Definition (DTD)}, the most widespread XML
schema formalism for (ordered) XML \cite{BeNeVa04,GrMa11}, is
essentially a set of rules associating with each label a regular
expression that defines the admissible sequences of children. The DTDs are
best fitted towards ordered content because they use regular
expressions, a formalism that defines sequences of labels. However,
when unordered content model needs to be defined, there is a tendency
to use \emph{over-permissive} regular expressions. For instance, the
DTD below corresponds to the one used in practice for the DBLP
repository:
\[
\begin{tabular}{rcl}
{\sl dblp} & $\rightarrow$ & ({\sl article $\mid$ book})$^*$\\
{\sl article} & $\rightarrow$ & ({\sl title $\mid$ year $\mid$ author})$^*$ \\
{\sl book} & $\rightarrow$ & ({\sl title $\mid$ year $\mid$ author $\mid$ editor $\mid$ publisher})$^*$
\end{tabular}
\]
This DTD allows an {\sl article} to contain any number of {\sl title},
{\sl year}, and {\sl author} elements. A {\sl book} may also have any
number of {\sl title}, {\sl year}, {\sl author}, {\sl editor}, and
{\sl publisher} elements. These regular expressions are clearly
over-permissive because they allow documents that do not follow the
intuitive guidelines set out earlier e.g., a document containing an
{\sl article} with two {\sl title}'s and no {\sl author} should not be
admissible.

While it is possible to capture unordered content models with regular
expressions, a simple pumping argument shows that their size may need
to be exponential in the number of possible labels of the children. In
case of the DBLP repository, this number reaches values up to 12,
which basically precludes any practical use of such regular
expressions. This suggests that over-permissive regular expressions
may be employed for the reasons of conciseness and readability.


The use of over-permissive regular expressions, apart from allowing
documents that do not follow the guidelines, has other negative
consequences e.g., in static analysis tasks that involve the
schema. Take for example the following two twig
queries~\cite{ACLS02,XPath1}:
\begin{align*}
&/\mbox{\sl dblp}/\mbox{\sl book}[\mbox{\sl author}=\mbox{``$\mathit{C.\ Papadimitriou}$''}]\\
&/\mbox{\sl dblp}/\mbox{\sl book}[\mbox{\sl author}=\mbox{``$\mathit{C.\ Papadimitriou}$''}][\mbox{\sl title}]
\end{align*}
The first query selects the elements labeled {\sl book}, children of
{\sl dblp} and having an {\sl author} containing the text
``\emph{C.\ Papadimitriou.}'' The second query additionally requires
that {\sl book} has a {\sl title}. Naturally, these two queries
should be equivalent because every {\sl book} element should have a
{\sl title} child. However, the DTD above does not capture
properly this requirement, and, consequently, the two queries are not
equivalent w.r.t.\ this DTD.

In this paper, we study two new schema formalisms: the
\emph{disjunctive multiplicity schema (DMS)} and its restriction, the
\emph{disjunction-free multiplicity schema (MS)}. While they use a
user-friendly syntax inspired by DTDs, they define unordered content
model only, and, therefore, they are better suited for unordered XML. A DMS is a
set of rules associating with each label the possible number of
occurrences for all the allowed children labels by using
\emph{multiplicities}: ``$*$'' (0 or more occurrences), ``$+$'' (1 or
more), ``$?$'' (0 or 1), ``$1$'' (exactly 1 occurrence; often omitted
for brevity).  Additionally, alternatives can be specified using
restricted \emph{disjunction} (``$\mid$'') and all the conditions are
gathered with \emph{unordered concatenation} (``$\shuffle$''). For
instance, the following DMS captures precisely the intuitive
requirements for the DBLP repository:
\[
\begin{tabular}{rcl}
  {\sl dblp} & $\rightarrow$ & {\sl article}$^* \shuffle$ {\sl book}$^*$\\
  {\sl article} & $\rightarrow$& {\sl title} $\shuffle$ {\sl year} $\shuffle$ {\sl author}$^+$\\
  {\sl book} & $\rightarrow$ & {\sl title} $\shuffle$ {\sl year} $\shuffle$ {\sl publisher}$^? \shuffle$ ({\sl author}$^+\mid$ {\sl editor}$^+$)
\end{tabular}
\]
In particular, an {\sl article} must have exactly one {\sl title},
exactly one {\sl year}, and at least one {\sl author}. A {\sl book}
may additionally have a {\sl publisher} and may have one or more {\sl
  editor}'s instead of {\sl author}'s. Note that, unlike the DTD
defined earlier, this DMS does not allow documents having an {\sl
  article} with several {\sl title}'s or without any {\sl author}.

There has been an attempt to use DTD-like rule based schemas to define
unordered content models by interpreting the regular expressions under
\emph{commutative closure}~\cite{BeMi99}: essentially, an unordered
collection of children matches a regular expression if there exists an
ordering that matches the regular expression in the standard
way. However, testing whether there exists a permutation of a word
that matches a regular expression is NP-complete~\cite{KoTo10}, which
implies a significant increase in computational complexity of the
membership problem i.e., validating an XML document against the
schema. The schema formalisms proposed in this paper, DMS and MS, can
be seen as DTDs interpreted under commutative closure using restricted
classes of regular expressions. Two natural questions arise: do these
restrictions allow us to avoid the increase in computational
complexity, and how much of the expressiveness of DTDs is retained. The
answers are generally positive. There is no increase in computational
complexity but also no decrease
(cf.~Table~\ref{tab:complexity_summary}). Furthermore, the proposed
schema formalisms seem to capture a significant part of the
expressiveness of DTDs used in practice
(Section~\ref{sec:expressiveness}).

\begin{table*}
\begin{changemargin}{-1.75cm}{0cm} 
\begin{center}
\begin{small}
\begin{tabular}{|l|l|l|l|l|}
  \hline
  \emph{Problem of interest} & $\DTD$ & $\dms$& disjunction-free $\DTD$ & $\ms$\\\hline
  Schema satisfiability & PTIME \cite{BuWo98,Schwentick04b} & PTIME (Prop. \ref{sat-memb-ptime}) & PTIME \cite{BuWo98,Schwentick04b} & PTIME (Prop. \ref{sat-memb-ptime})  \\\hline

  Membership & PTIME \cite{BuWo98,Schwentick04b} & PTIME (Prop. \ref{sat-memb-ptime}) & PTIME \cite{BuWo98,Schwentick04b} & PTIME (Prop. \ref{sat-memb-ptime}) \\\hline
  Schema containment& PSPACE-c${}^\dag$/PTIME \cite{BuWo98,Schwentick04b} & PTIME (Th. \ref{cnt-ptime}) & coNP-h${}^\dag$/PTIME \cite{BuWo98,MaNeSc09} & PTIME (Th. \ref{cnt-ptime}) \\\hline

Query satisfiability${}^\ddag$ & NP-c \cite{BeFaGe08} & NP-c (Prop. \ref{satq-nphard}) & PTIME \cite{BeFaGe08} & PTIME (Th. \ref{sat-impl-ptime}) \\\hline

Query implication${}^\ddag$ & EXPTIME-c \cite{NeSc06}& EXPTIME-c (Prop. \ref{exptime-c})&PTIME (Cor. \ref{cor})& PTIME (Th. \ref{sat-impl-ptime})\\\hline
Query containment${}^\ddag$ & EXPTIME-c \cite{NeSc06} &EXPTIME-c (Prop. \ref{exptime-c}) &coNP-c (Cor. \ref{cor})&coNP-c (Th. \ref{th:cnt-ms-hard})  \\\hline
\multicolumn{5}{l}{{${}^\dag$ when non-deterministic regular expressions are used. ${}^\ddag$ for twig queries.}}
\end{tabular}
\end{small}
\end{center}
\vspace{-1pt}
\caption{\label{tab:complexity_summary}Summary of complexity results.}
\end{changemargin}
\end{table*}

We study the complexity of several basic decision problems: schema
satisfiability, membership of a tree to the language of a schema,
containment of two schemas, twig query satisfiability, implication,
and containment in the presence of
schema. Table~\ref{tab:complexity_summary} contains the summary of
complexity results compared with general DTDs and disjunction-free
DTDs. The lower bounds for the decision problems for DMS and MS are
generally obtained with easy adaptations of their counterparts for
general DTDs and disjunction-free DTDs. To obtain upper bounds we
develop several new tools. \emph{Dependency graphs} for MS and a
generalized definition of an \emph{embedding} of a query help us to
reason about query satisfiability, query implication, and query
containment in the presence of MS. An alternative characterization of
DMS with \emph{characterizing triples} is used to reduce the containment
of DMS to the containment of their characterizing triples, which can be
tested in PTIME. We add that our constructions and results for MS
extend easily to disjunction-free DTDs and allow to solve the problems
of query implication and query containment, which, to the best of our
knowledge, have not been previously studied for disjunction-free DTDs.


\noindent {\bf Related work.} Languages of unordered trees can 
be expressed by \emph{logic formalisms} or by \emph{tree automata}. 
Boneva et al.\ \cite{BT05, BTT05} make a survey on such formalisms and 
compare their expressiveness. The fundamental difference resides in
the  kind of constraints that can be expressed
for the allowed collections of children for some node. We mention here
only formalisms introduced in the context of XML.
{\em Presburger automata} \cite{SeScMu03}, 
{\em sheaves automata} \cite{DaLu03}, and the {\em TQL logic} \cite{CG04} 
allow to express {\em Presburger constraints} on the 
numbers of occurrences of the different symbols among the children 
of some node. This is also equivalent to considering DTDs under 
commutative closure, similarly to \cite{BeMi99}. The consequence of the high 
expressive power is that the
membership problem is NP-complete for an unbounded alphabet \cite{KoTo10}. 
Therefore, these formalisms were not extensively used in practice. 
Suitable restrictions on Presburger automata
and on the TQL logic allow to obtain the same expressiveness as
the MSO logic on unordered trees \cite{BT05, BTT05}. 
DMS are strictly less expressive than these MSO-equivalent languages. Static analysis problems involving
twig queries were not studied for these languages. Additionally, we believe that DMS are more appropriate to be used as schema languages, as they were designed as such, in particular regarding the more user-friendly DTD-like syntax.
As mentioned earlier, unordered content model can also be defined by DTDs defining 
commutatively-closed sets of ordered trees. An (ordered) 
tree matches such a DTD iff all tree obtained by reordering of sibling
nodes also matches the DTD. This also turns out to be equally expressive as MSO on unordered trees  \cite{BT05, BTT05}. 
However, such a DTD may be of exponential size w.r.t.\ the size of the alphabet and, moreover, it is
PSPACE-complete to test whether a DTD defines a commutatively-closed set of trees~\cite{NeSc99}, which makes such 
DTDs unusable in practice.
XML Schema allow for a bounded number of symbols to appear in
arbitrary order, and RELAX NG allows to interleave sequences of
symbols of bounded length.  In contrast, the Kleene star in DMS allows
for unbounded unordered collections of children.  Schematron allows to
specify very general constraints on the number of occurrences of
symbols among the children of a node, in particular Presburger
constraints are expressible.  Schema languages using regular
expressions with unbounded interleaving were studied in
\cite{GeMaNe09}. These are more expressive than DMS but exhibit high
computational complexity of inclusion \cite{GeMaNe09} and membership
\cite{BeBjHo11}. To the best of our knowledge, the static analysis
problems involving queries were not studied for these languages when
unordered content is allowed.

\noindent {\bf Organization.} The paper is organized as follows. 
In Section~\ref{sec:prelim} we introduce some preliminary notions, while in Section~\ref{sec:schemas} we present our schema formalisms.
In Section~\ref{sec:static} we define the problems of interest and then we analyze them for DMS (Subsection~\ref{subsec:dms}) and for MS (Subsection~\ref{subsec:ms}).
In Section~\ref{sec:expressiveness} we discuss the expressiveness of the proposed formalisms, while in Section~\ref{sec:conclusions} we summarize our results and outline further directions.
Because of space restriction, we present only sketches of some proofs; complete proofs will be given in the full version of the paper, which is currently in preparation for journal submission.

\section{Preliminaries}\label{sec:prelim}
\noindent Throughout this paper we assume an alphabet $\Sigma$ which
is a finite set of symbols.

\noindent {\bf Trees.} We model XML documents with unordered labeled
trees.  Formally, a {\em tree} $t$ is a tuple
$(N_t,\root_t,\lab_t,\child_t)$, where $N_t$ is a finite set of nodes,
$\root_t\in N_t$ is a distinguished root node,
$\lab_t:N_t\rightarrow\Sigma$ is a labeling function, and
$\child_t\subseteq N_t\times N_t$ is the parent-child relation.  We
assume that the relation $\child_t$ is acyclic and require every
non-root node to have exactly one predecessor in this relation.  By
$\Tree$ we denote the set of all finite trees.
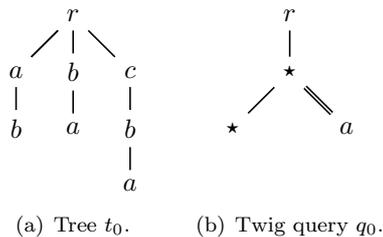
\begin{figure}[htb]
  \centering
  \subfigure[Tree $t_0$.]{\label{fig:tree}
    \centering
    \begin{tikzpicture}[yscale=0.75]
      \path[use as bounding box] (-1.25,.25) rectangle (1.25,-3.25);
      \node at (0,0) (n0) {$r$};
      \node at (-0.75,-1) (n1) {$a$};
      \node at (0,-1) (n4) {$b$};
      \node at (0,-2) (n5) {$a$};
      \node at (0.75,-1) (n2) {$c$};
      \node at (0.75,-2) (n3) {$b$};
      \node at (0.75,-3) (n6) {$a$};
      \node at (-0.75,-2) (n7) {$b$};
      \draw[-,semithick] (n1) -- (n7);      
      \draw[-,semithick] (n0) -- (n1);
      \draw[-,semithick] (n0) -- (n1);
      \draw[-,semithick] (n0) -- (n2);
      \draw[-,semithick] (n2) -- (n3);
      \draw[-,semithick] (n3) -- (n6);
      \draw[-,semithick] (n0) -- (n4);
      \draw[-,semithick] (n4) -- (n5);
    \end{tikzpicture}
  }
  \subfigure[Twig query $q_0$.]{\label{fig:twig0}
    \centering
    \begin{tikzpicture}[yscale=0.75]
      \path[use as bounding box] (-1.25,.25) rectangle (1.25,-3.25);
      \node at (0,0) (n0) {$r$};
      \node at (0,-1) (n2) {$\wc$};
      \node at (0.75,-2) (n3) {$a$};
      \node at (-0.75,-2) (n1) {$\wc$};
      \draw[-,semithick] (n2) -- (n1);
      \draw[-,semithick] (n0) -- (n2);
      \draw[-,double,semithick] (n2) -- (n3);
    \end{tikzpicture}
  }
  \caption{A tree and a twig query.}
  \label{fig:trees-twigs}
\end{figure}

\noindent {\bf Queries.}  We work with the class of twig queries,
which are essentially unordered trees whose nodes may be additionally
labeled with a distinguished wildcard symbol $\wc\not\in\Sigma$ and
that use two types of edges, child ($/$) and descendant ($\dblslash$),
corresponding to the standard XPath axes.  Note that the semantics of
$\dblslash$-edge is that of a proper descendant (and not that of
descendant-or-self).  Formally, a \emph{twig query} $q$ is a tuple
$(N_q,\root_q,\lab_q,\child_q,\desc_q)$, where $N_q$ is a finite set
of nodes, $\root_q\in N_q$ is the root node,
$\lab_q:N_q\rightarrow\Sigma\cup\{\wc\}$ is a labeling function,
$\child_q\subseteq N_q\times N_q$ is a set of child edges, and
$\desc_q\subseteq N_q\times N_q$ is a set of descendant edges.  We
assume that $\child_q\cap\desc_q=\emptyset$ and that the relation
$\child_q\cup\desc_q$ is acyclic and we require every non-root node to
have exactly one predecessor in this relation.  By $\Twig$ we denote
the set of all twig queries. Twig queries are often presented using
the abbreviated XPath syntax~\cite{XPath1} e.g., the query
$q_0$ in Figure~\ref{fig:twig0} can be written as
$r/\wc[\wc]\dblslash{}a$.

\noindent {\bf Embeddings.} We define the semantics of twig queries
using the notion of embedding which is essentially a mapping of nodes
of a query to the nodes of a tree that respects the semantics of the
edges of the query.  Formally, for a query $q\in \Twig$ and a tree
$t\in \Tree$, an \emph{embedding} of $q$ in $t$ is a function $\lambda
: N_q \rightarrow N_t$ such that:
\begin{enumerate}
\itemsep0pt
\item[$1$.] $\lambda(\root_q)=\root_t$,
\item[$2$.] for every $(n,n')\in \child_q$,
  $(\lambda(n),\lambda(n'))\in \child_t$,
\item[$3$.] for every $(n,n')\in \desc_q$,
  $(\lambda(n),\lambda(n'))\in( \child_t)^+$ (the transitive closure of $\child_t$),
\item[$4$.] for every $n\in N_q$, $\lab_q(n) = \wc$ or $\lab_q(n) = \lab_t(\lambda(n))$.
\end{enumerate}
If there exists an embedding from $q$ to $t$ we say that $t$
\emph{satisfies} $q$ and we write $t\models q$.  By $ L(q)$ we
denote the set of all the trees satisfying $q$.  Note that we do not
require the embedding to be injective i.e., two nodes of the query may
be mapped to the same node of the tree.  Figure~\ref{fig:embeddings}
presents all embeddings of the query $q_0$ in the tree $t_0$ from
Figure~\ref{fig:trees-twigs}.
\begin{figure}[htb]
  \centering
  \begin{tikzpicture}[yscale=0.75]
    \node at (0,0) (n0) {$r$};
    \node at (-0.75,-1) (n1) {$a$};
    \node at (0,-1) (n4) {$b$};
    \node at (0,-2) (n5) {$a$};
    \node at (0.75,-1) (n2) {$c$};
    \node at (0.75,-2) (n3) {$b$};
    \node at (0.75,-3) (n6) {$a$};
    \node at (-0.75,-2) (n7) {$b$};
    \draw[-,semithick] (n1) -- (n7);      
    \draw[-,semithick] (n0) -- (n1);
    \draw[-,semithick] (n0) -- (n2);
    \draw[-,semithick] (n2) -- (n3);
    \draw[-,semithick] (n3) -- (n6);
    \draw[-,semithick] (n0) -- (n4);
    \draw[-,semithick] (n4) -- (n5);
    \begin{scope}[xshift=-2.75cm]
      \node at (0,0) (m0) {$r$};
      \node at (0,-1) (m2) {$\wc$};
      \node at (0.75,-2) (m3) {$a$};
      \node at (-0.75,-2) (m1) {$\wc$};
      \draw[-,semithick] (m2) -- (m1);
      \draw[-,semithick] (m0) -- (m2);
      \draw[-,double,semithick] (m2) -- (m3);
    \draw (m0) edge[->,bend left,densely dotted] (n0);
    \draw (m1) edge[->,bend right,densely dotted] (n5);
    \draw (m2) edge[->,bend left,densely dotted] (n4);
    \draw (m3) edge[->,bend left,densely dotted] (n5);
    \end{scope}
    \begin{scope}[xshift=2.75cm]
      \node at (0,0) (m0) {$r$};
      \node at (0,-1) (m2) {$\wc$};
      \node at (0.75,-2) (m3) {$a$};
      \node at (-0.75,-2) (m1) {$\wc$};
      \draw[-,semithick] (m2) -- (m1);
      \draw[-,semithick] (m0) -- (m2);
      \draw[-,double,semithick] (m2) -- (m3);
    \draw (m0) edge[->,bend right,densely dotted] (n0);
    \draw (m1) edge[->,bend right,densely dotted] (n3);
    \draw (m2) edge[->,bend right,densely dotted] (n2);
    \draw (m3) edge[->,bend left,densely dotted] (n6);
    \end{scope}
  \end{tikzpicture}
  \caption{Embeddings of $q_0$ in $t_0$.}
  \label{fig:embeddings}
\end{figure}
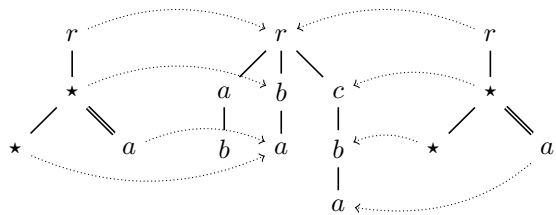

\noindent {\bf Unordered words.} An \emph{unordered word} is
essentially a multiset of symbols i.e., a function
$w:\Sigma\rightarrow\mathbb N_0$ mapping symbols from the alphabet to
natural numbers, and we call the number $w(a)$ the number of
occurrences of the symbol $a$ in $w$. We also write $a\in w$ as a
shorthand for $w(a)\neq 0$. An empty word $\varepsilon$ is an
unordered word that has $0$ occurrences of every symbol i.e.,
$\varepsilon(a)=0$ for every $a\in\Sigma$. We often use a simple
representation of unordered words, writing each symbol in the alphabet
the number of times it occurs in the unordered word. For example, when
the alphabet is $\Sigma=\{a,b,c\}$, $w_0=aaacc$ stands for the
function $w_0(a) = 3$, $w_0(b) = 0$, and $w_0(c) = 2$. 

The (unordered) concatenation of two unordered words $w_1$ and $w_2$
is defined as the multiset union $w_1\uplus w_2$ i.e., the function
defined as $(w_1\uplus w_2)(a) = w_1(a)+w_2(a)$ for all
$a\in\Sigma$. For instance, $aaacc\uplus{}abbc=aaaabbccc$. Note that
$\varepsilon$ is the identity element of the unordered concatenation
$\varepsilon\uplus w = w\uplus \varepsilon = w$ for all unordered word
$w$. Also, given an unordered word $w$, by $w^i$ we denote the
concatenation $w\uplus\ldots\uplus w$ ($i$ times).

A \emph{language} is a set of unordered words. The unordered
concatenation of two languages $L_1$ and $L_2$ is a language
$L_1\uplus L_2 = \{w_1\uplus w_2\mid w_1\in L_1, w_2\in L_2\}$. For
instance, if $L_1 = \{a, aac\}$ and $L_2 = \{ac, b, \varepsilon\}$,
then $L_1\uplus L_2 = \{a,ab,aac,aabc,aaacc\}$.


\section{Multiplicity schemas}\label{sec:schemas}
A \emph{multiplicity} is an element from the set $\{*,+,?,0,1\}$.
\noindent We define the function $\llbracket\cdot\rrbracket$ mapping
multiplicities to sets of natural numbers.  More precisely:
\[
\begin{tabular}{ccccc}
$\llbracket*\rrbracket = \{0,1,2,\ldots\}$, & 
$\llbracket+\rrbracket = \{1,2,\ldots\}$, & 
$\llbracket?\rrbracket = \{0,1\}$, &
$\llbracket1\rrbracket = \{1\}$, & 
$\llbracket0\rrbracket = \{0\}$. 
\end{tabular}
\]
Given a symbol $a\in\Sigma$ and a multiplicity $M$, the language of
$a^M$, denoted $ L(a^M)$, is $\{a^i\mid i\in\llbracket
M\rrbracket\}$.  For example, $ L(a^+) = \{a,aa,\ldots\}$,
$ L(b^0) = \{\varepsilon\}$, and $
L(c^?)=\{\varepsilon,c\}$.

A \emph{disjunctive multiplicity expression} $E$ is:
\[
E \colonequals D_1^{M_1} \shuffle \ldots \shuffle D_n^{M_n},
\]
where for all $1\leq i \leq n$, $M_i$ is a multiplicity and each $D_i$ is:
\[
D_i \colonequals a_1^{M_1'} \mid \ldots \mid a_k^{M_k'},
\]
where for all $1\leq j \leq k$, $M_j'$ is a multiplicity and
$a_j\in\Sigma$.  Moreover, we require that every symbol $a\in \Sigma$
is present at most once in a disjunctive multiplicity expression.  For
instance, $(a\mid b)\shuffle(c\mid d)$ is a disjunctive multiplicity
expression, but $(a\mid b)\shuffle c\shuffle(a\mid d)$ is not because
$a$ appears twice. A \emph{disjunction-free multiplicity expression}
is an expression which uses no disjunction symbol ``$\mid$'' i.e., an
expression of the form $a_1^{M_1} \shuffle \ldots \shuffle a_k^{M_k}$,
where for all $1\leq i \leq k$, the $a_i$'s are pairwise distinct
symbols in the alphabet and the $M_i$'s are multiplicities.

The language of a disjunctive multiplicity expression is:
\begin{gather*}
 L(a_1^{M_1} \mid \ldots \mid a_k^{M_k}) =  L(a_1^{M_1}) \cup \ldots\cup L(a_k^{M_k}),\\
 L(D^M) = \{w_1\uplus\ldots\uplus w_i\mid w_1,\ldots,w_i\in L(D)\wedge i\in\llbracket M\rrbracket\},\\
 L(D_1^{M_1} \shuffle \ldots \shuffle D_n^{M_n}) =  L(D_1^{M_1}) \uplus \ldots\uplus  L(D_n^{M_n}).  
\end{gather*}
If an unordered word $w$ belongs to the language of a disjunctive multiplicity expression $E$, we denote it $w\models E$.
When a symbol $a$ (resp. a disjunctive multiplicity expression $E$)
has multiplicity $1$, we often write $a$ (resp. $E$) instead of $a^1$
(resp. $E^1$).  Moreover, we omit writing symbols and disjunctive
multiplicity expressions with multiplicity $0$. Take for instance,
$E_0= a^+ \shuffle (b \mid c) \shuffle d^?$ and note that both the
symbols $b$ and $c$ as well as the disjunction $(b \mid c)$ have an
implicit multiplicity $1$. The language of $E_0$ is: 
\[
 L(E_0) = \{a^ib^jc^kd^\ell \mid 
i,j,k,\ell \in \mathbb{N}_0,\ 
i \geq 1,\ 
j+k = 1,\ 
\ell \leq 1
\}.
\]
Next, we formally define the proposed schema formalisms.
\begin{definition}
  A \emph{disjunctive multiplicity schema (DMS)} is a tuple
  $S=(\root_S, R_S)$, where $\root_S\in\Sigma$ is a designated root
  label and $R_S$ maps symbols in $\Sigma$ to disjunctive multiplicity
  expressions.  By $\dms$ we denote the set of all disjunctive
  multiplicity schemas.  A \emph{disjunction-free multiplicity schema
    (MS)} $S=(\root_S, R_S)$ is a restriction of the $\dms$, where
  $R_S$ maps symbols in $\Sigma$ to disjunction-free multiplicity
  expressions.  By $\ms$ we denote the set of all disjunction-free
  multiplicity schemas.
\end{definition}
To define satisfiabily of a DMS (or MS) $S$ by a tree $t$ we first
define the unordered word $ch_t^n$ of children of a node $n\in N_t$ of
$t$ i.e., $\ch_t^n(a) = |\{m\in N_t\mid(n,m)\in\child_t\wedge
\lab_t(m)=a\}|$. Now, a tree $t$ \emph{satisfies} $S$, in symbols $t
\models S$, if $\lab_t(\root_t) = \root_S$ and for any node $n\in
N_t$, $\ch_t^n\in L(R_S(\lab_t(n)))$. By $
L(S)\subseteq \Tree$ we denote the set of all the trees satisfying
$S$. 

In the sequel, we represent a schema $S=(\root_S,R_S)$ as a set of
rules of the form $a\rightarrow R_S(a)$, for any $a\in\Sigma$.  If
$ L(R_S(a)) = \varepsilon$, then we write $a\rightarrow
\epsilon$ or we simply omit writing such a rule.
\begin{example}\normalfont We present schemas $S_1,S_2,S_3,S_4$ illustrating the
  formalisms defined above.  They have the root label $r$ and the
  rules:
\begin{align*}
S_1&:~~~r\rightarrow a\shuffle b^*\shuffle c^?&a&\rightarrow b^?&b& \rightarrow a^?&c&\rightarrow b\\[-4pt]
S_2&:~~~r\rightarrow c\shuffle b\shuffle a&a&\rightarrow b^?&b& \rightarrow a&c&\rightarrow b\\[-4pt]
S_3&:~~~r\rightarrow (a\mid b)^+\shuffle c&a&\rightarrow b^?&b& \rightarrow a^?&c&\rightarrow b\\[-4pt]
S_4&:~~~r\rightarrow (a\mid b\mid c)^*&a&\rightarrow \epsilon&b& \rightarrow a^?&c&\rightarrow b
\end{align*}
$S_1$ and $S_2$ are $\ms$, while $S_3$ and $S_4$ are $\dms$.
The tree $t_0$ from Figure~\ref{fig:tree} satisfies only $S_1$ and $S_3$.
\qed\end{example}

\section{Static analysis}\label{sec:static}
We first define the problems of interest and we formally state the
corresponding decision problems parameterized by the class of schema
and, when appropriate, by a class of queries.

\noindent\textbf{Schema satisfiability} --
checking if there exists a tree satisfying the given schema:
\[
\SAT_{\mathcal{S}} = \{ S \in\mathcal{S} \mid \exists t\in
\Tree.\ t\models S \}.
\]
\noindent {\bf Membership} -- checking if the given tree satisfies the
given schema:
\[
\MEMB_{\mathcal{S}} = 
\{(S,t) \in\mathcal{S}\times\Tree\mid t\models S
\}.
\]
\noindent {\bf Schema containment} -- checking if every tree
satisfying one given schema satisfies another given schema:
\[
\CNT_{\mathcal{S}} = \{(S_1,
S_2)\in\mathcal{S}\times\mathcal{S}\mid L(S_1)\subseteq
 L(S_2)\}.
\]
\noindent {\bf Query satisfiability by schema} -- checking if
there exists a tree that satisfies the given schema and the given
query:
\[\SAT_{\mathcal{S},\mathcal{Q}} = \{(S,q)\in
\mathcal{S}\times\mathcal{Q}\mid \exists t\in L(S).\ t\models
q\}.
\]
\noindent {\bf Query implication by schema} -- checking if every tree satisfying
the given schema satisfies also the given query:
\[
\IMPL_{\mathcal{S},\mathcal{Q}} = \{(S, q)\in
\mathcal{S}\times\mathcal{Q}\mid \forall t\in L(S).\ t\models
q\}.
\]
\noindent {\bf Query containment in the presence of schema} -- checking if every
tree satisfying the given schema and one given query also satisfies
another given query:
\[
\CNT_{\mathcal{S},\mathcal{Q}} = \{(p, q, S) \in
\mathcal{Q}\times\mathcal{Q}\times\mathcal{S} \mid \forall
t\in L(S).\ t \models p\Rightarrow t\models q\}.
\]
We next study these decision problems for DMS an MS.

\subsection{Disjunctive multiplicity schema}\label{subsec:dms}
In this subsection we present the static analysis for DMS. 
We first introduce the notion of normalized disjunctive multiplicity expressions and an alternative definition with characterizing triples.
Finally, we state the complexity results for DMS.
\subsubsection{Normalized disjunctive multiplicity expressions}
Recall that a disjunctive multiplicity expression has the form $E=D_1^{M_1}\shuffle\ldots\shuffle D_m^{M_m}$.
Intuitively, in a \emph{normalized} disjunctive multiplicity expression, every disjunction $D_i^{M_i}$ has one of the following three forms:
\begin{enumerate}
\item $(a_1\mid\ldots\mid a_n)^+$,
\item $(a_1^{M_1}\mid\ldots\mid a_n^{M_n})$, where $\forall j.\ 1\leq j\leq n.\ 0\notin\llbracket M_j\rrbracket$,
\item $(a_1^{M_1}\mid\ldots\mid a_n^{M_n})$, where $\forall j.\ 1\leq j\leq n.\ 0\in\llbracket M_j\rrbracket$.
\end{enumerate}
Given a disjunctive multiplicity expression $E=D_1^{M_1}\shuffle\ldots\shuffle D_m^{M_m}$, we denote by $\Sigma_{D_i}$ the set of symbols used in the disjunction from $D_i$ and by $M^a$ the multiplicity corresponding to a symbol $a$.
Formally, we say that $E$ is \emph{normalized} if the following two conditions are satisfied:
\begin{flalign*}
&\forall i.\ 1\leq i\leq m.\ M_i\neq1\Rightarrow M_i=+ \wedge \forall a\in \Sigma_{D_i}.\ M^a=1,\\
&\forall i.\ 1\leq i\leq m.\ (\exists a\in \Sigma_{D_i}.\ 0\in \llbracket M^a\rrbracket) \Rightarrow (\forall a'\in\Sigma_{D_i}.\ 0\in \llbracket M^{a'}\rrbracket).
\end{flalign*}
Any $D_i^{M_i}$ can be rewritten as an equivalent normalized disjunctive multiplicity expression using the following rules:
\begin{itemize}
\item $(a_1^{M_1}\mid\ldots\mid a_n^{M_n})^*$ goes to $a_1^*\shuffle\dots\shuffle a_n^*$.
\item $(a_1^{M_1}\mid\ldots\mid a_n^{M_n})^?$ goes to $(a_1^{M_1'}\mid\ldots\mid a_n^{M_n'})$, where $\forall j.\ 1\leq j\leq n.\ \llbracket M_j' \rrbracket = \{0\}\cup\llbracket M_j \rrbracket$.
\item $(a_1^{M_1}\mid\ldots\mid a_n^{M_n})$, where $\exists j.\ 1\leq j\leq n.\ 0\in \llbracket M_j\rrbracket $ goes to $(a_1^{M_1'}\mid\ldots\mid a_n^{M_n'})$, where $\forall j.\ 1\leq j\leq n.\ \llbracket M_j' \rrbracket = \{0\}\cup\llbracket M_j \rrbracket$.
\item $(a_1^{M_1}\mid\ldots\mid a_n^{M_n})^+$, where $\exists j.\ 1\leq j\leq n.\ 0\in \llbracket M_j\rrbracket $ goes to $a_1^*\shuffle\dots\shuffle a_n^*$.
\item $(a_1^{M_1}\mid\ldots\mid a_n^{M_n})^+$, where $\forall j.\ 1\leq j\leq n.\ 0\notin \llbracket M_j\rrbracket $ goes to $(a_1\mid\ldots\mid a_n)^+$.
\item $(a_1^{M_1}\mid\ldots\mid a_n^{M_n})^0$ is removed.
\item $a^0$ occurring in some disjunction is removed.
\end{itemize}
Note that each of the rewriting steps gives an equivalent expression.
From now on, we assume w.l.o.g.\ that all the disjunctive multiplicity expressions that we manipulate are normalized.

\subsubsection{Alternative definition with characterizing triples}
We propose an alternative definition of the
language of a disjunctive multiplicity expression using a
\emph{characterizing triple}.
Moreover, we show that each element of the triple has a compact representation which is polynomial in the size of the alphabet and computable in PTIME.
Recall that the disjunctive multiplicity expressions do not allow repetitions of symbols hence they have \emph{size} linear in $|\Sigma|$.
Next, we prove that the inclusion of two disjunctive multiplicity expressions is equivalent to the inclusion of the characterizing triples.
Thus, we can view the characterizing triple as a \emph{normal form} of a disjunctive multiplicity expression.
Recall that $a\in w$ means that $w(a)\neq 0$. 

Given a disjunctive multiplicity expression $E$, we define the \emph{characterizing triple} $(C_E,N_E,P_E)$ consisting of the following sets:
\begin{itemize}
\item The \emph{conflicting pairs of siblings} $C_E$ consists of
  pairs of symbols in $\Sigma$ such that $E$ defines no word using
  both symbols simultaneously:
  \[
  C_E = \{(a_1,a_2)\in\Sigma\times\Sigma\mid \varnot\exists
  w\in L(E).\ a_1\in w\wedge a_2\in w\}.
  \]
\item The \emph{extended cardinality map} $N_E$ captures for each symbol
 in the alphabet the possible numbers of its occurrences in
  the unordered words defined by $E$:
  \[
  N_{E} = \{(a, w(a))\in \Sigma\times\mathbb N_0\mid w\in  L(E)\}.
  \]
\item The \emph{sets of required symbols} $P_E$ which captures symbols
  that must be present in every word; essentially, a set of symbols $X$
  belongs to $P_E$ if every word defined by $E$ contains at least one
  element from $X$:
  \[
  P_E=\{X\subseteq\Sigma\mid\forall w\in L(E).\ \exists a\in
  X.\ a \in w\}.
  \]
\end{itemize}
As an example we take $E_0= a^+ \shuffle (b \mid c) \shuffle d^?$.
Because $P_E$ is closed under supersets, we list only its minimal
elements:
\begin{gather*}
  C_{E_0} = \{ (b,c), (c,b) \}, \qquad 
  P_{E_0} = \{ \{a\}, \{b,c\}, \ldots\},\\
  N_{E_0} = \{ (b,0), (b,1), (c,0), (c,1), (d,0), (d,1), (a,1), (a,2), 
  \ldots\}.\\
\end{gather*}
An unordered word $w$ is consistent with the triple $(C_E, N_E, P_E)$ corresponding to a disjunctive multiplicity expression $E$, denoted $w\models (C_E, N_E, P_E)$ if $w$ is consistent with $C_E, N_E$, and $P_E$, respectively.
Formally:
\begin{flalign*}
&w\models C_E\colonequals\forall (a_1, a_2) \in C_E.\ (a_1\in w \Rightarrow a_2 \notin w) \wedge (a_2\in w \Rightarrow a_1\notin w),\\
&w\models N_E\colonequals\forall a \in \Sigma.\ (a,w(a))\in N_E,\\
&w\models P_E\colonequals\forall X\in P_E.\ \exists a\in X.\ a\in w.
\end{flalign*}
Furthermore, each element of a characterizing triple has a \emph{compact representation}, which is polynomial in the size of the alphabet and computable in PTIME.
Next, we present the construction for each compact representation:
\begin{itemize}
\item Given a disjunctive multiplicity expression $E=D_1^{M_1}
\shuffle\dots\shuffle D_m^{M_m}$, the size of $C_E$ is quadratic in $|\Sigma|$, but we can represent it linearly in $|\Sigma|$.
Thus, we obtain $C_E^*$, which consists intuitively of non-singleton sets of labels from the same disjunction from $E$ such that the multiplicity associated to the disjunction is $1$:
\[
C_E^*=\{X\subseteq\Sigma_{D_1}\cup\ldots\cup\Sigma_{D_m}\mid\forall a,a'\in X.\ a\neq a'\Rightarrow (a,a')\in C_E\}.
\]
Then, $(a,b) \in \Sigma \times \Sigma$ belongs to $C_E$ iff one of the following holds: (i) there exists $X \in C_E^*$ s.t. $\{a,b\} \subseteq X$, or (ii)  $a \not\in \Sigma_{D_1}\cup\ldots\cup\Sigma_{D_m}$  or  $b \not\in \Sigma_{D_1}\cup\ldots\cup\Sigma_{D_m}$.
\item Given a disjunctive multiplicity expression $E=D_1^{M_1}\shuffle\dots\shuffle D_m^{M_m}$, note that the set $N_E$ may be infinite, but it can be represented in a compact manner using multiplicities: for any label $a$, the set $\{x\in \mathbb N_0\mid(a,x)\in N_E\}$ is representable by a multiplicity.
Given a symbol $a\in\Sigma$, by $N_E^*(a)$ we denote the multiplicity $M$ such that $\llbracket M\rrbracket=\{x\in\mathbb N_0\mid (a,x)\in N_E\}$.
Moreover, for any $a\in\Sigma$, the multiplicity $N_E^*(a)$ can be easily obtained from $E$.
More precisely:
\[
N_E^*(a) = \begin{cases}
0,\ \textrm{if }\forall i.\ 1\leq i\leq m.\ a\notin\Sigma_{D_i},\\
M^a,\ \textrm{if }\exists i.\ 1\leq i\leq m .\ \Sigma_{D_i}=\{a\},\\
?,\ \textrm{if }\exists i.\ 1\leq i\leq m .\ a\in\Sigma_{D_i}\wedge M_i=1\wedge M^a\in\{?,1\},\\
*,\ \textrm{otherwise}.
\end{cases}
\]
Then, obviously, $(a,x) \in N_E$ iff $x \in \llbracket N^*_E(a) \rrbracket$.
\item $P_{E}$ may be exponential in $|\Sigma|$, but it can be represented with its $\subseteq$-\emph{minimal} elements:
\[
P_E^*=\{X\in P_E\mid \varnot\exists X'\in P_{E}.\ X'\subset X\}.
\]
For a disjunctive multiplicity expression $E=D_1^{M_1}\shuffle\dots\shuffle D_m^{M_m}$, $P_E^*$ consists intuitively of the disjunctions from $E$ such that the labels from the disjunction have multiplicities not accepting 0 occurrences.
Therefore, we can construct $P_E^*$ in a straightforward manner:
\[
P_E^* = \{\Sigma_{D_i}\mid 1\leq i\leq m\wedge\forall a\in\Sigma_{D_i}.\ 0\notin\llbracket M^a\rrbracket\}.
\]
Then $X \in P_E$ iff there exists $X' \in P_E^*$ s.t. $X' \subseteq X$.
\end{itemize}
For example, for the same $E_0= a^+ \shuffle (b \mid c) \shuffle d^?$, we have:
\begin{gather*}
C_{E_0}^* = \{\{b,c\}\}, \qquad P_{E_0}^* = \{\{a\}, \{b,c\}\},\\
N_{E_0}^*(a) = +, \qquad N_{E_0}^*(b) = N_{E_0}^*(c) = N_{E_0}^*(d) = ?.
\end{gather*}
We also illustrate the construction of the compact representation of the characterizing triple on a more complex disjunctive multiplicity expression:
\[
E_1 = (a\mid b)^+\shuffle (c^?\mid d^*\mid e^*)\shuffle f^+\shuffle g^?\shuffle (h^+ \mid i)
\]
over the alphabet $\Sigma=\{a,b,c,d,e,f,g,h,i,j\}$. 
We obtain:
\begin{gather*}
C_{E_1}^* = \{\{c,d,e\}, \{h,i\}\},\\
P_{E_1}^*=\{\{a,b\},\{f\},\{h,i\}\},\\
N_{E_1}^*(a) = N_{E_1}^*(b) = N_{E_1}^*(d) = N_{E_1}^*(e) = N_{E_1}^*(h) = *,\\
N_{E_1}^*(c) = N_{E_1}^*(g) = N_{E_1}^*(i) = ?, ~N_{E_1}^*(f) = +, ~N_{E_1}^*(j) = 0.
\end{gather*}
We use the characterizing triple to give an alternative characterization of the membership of an unordered word to the language of a disjunctive multiplicity expression:
\begin{lemma}\label{lemma:cons-tuple}
An unordered word $w$ belongs to the language of a disjunctive multiplicity expression $E$ iff it is consistent with the triple $(C_E, N_E, P_E)$.
\end{lemma}
\begin{proof}\normalfont
For the \emph{if} part, consider the triple $(C_E, N_E, P_E)$ corresponding to a normalized disjunctive multiplicity expression $E=D_1^{M_1}\shuffle\ldots\shuffle D_m^{M_m}$, and an unordered word $w$ such that $w\models (C_E, N_E, P_E)$.
Let $w=w_1\uplus\ldots\uplus w_m\uplus w'$, where, intuitively, each $w_i$ contains all the occurrences in $w$ of the symbols from $\Sigma_{D_i}$. 
Formally: 
\[
\forall i.\ 1\leq i\leq m.\ ((\forall a\in \Sigma_{D_i}.\ w_i(a) = w(a))\wedge (\forall a'\in\Sigma~\backslash~\Sigma_{D_i}.\ w_i(a') = 0)).
\]
Since $w\models N_E$, we infer that $\forall a\in\Sigma~\backslash~(\Sigma_{D_1}\cup\ldots\cup\Sigma_{D_m}).\ w(a)= 0$, which implies that $w'=\varepsilon$.
Thus, proving $w\models E$ reduces to proving that $\forall i.\ 1\leq i\leq m.\ w_i\models D_i^{M_i}$.
We prove while reasoning on each of the three possible forms of the disjunctions $D_i^{M_i}$, for every $i$ such that $1\leq i\leq m$:
\begin{enumerate}
\item $D_i^{M_i} = (a_1\mid\ldots\mid a_n)^+$, which implies that $\{a_1,\ldots,a_n\}\in P_E$.
Since $w$ is consistent with $P_E$, we infer that $\exists j.\ 1\leq j\leq n.\ a_j\in w$.
From the construction of $w_i$ we obtain $a_j\in w_i$, hence $w_i\models D_i^{M_i}$.
\item $D_i^{M_i} = (a_1^{M_1}\mid\ldots\mid a_n^{M_n})\wedge \forall j.\ 1\leq j\leq n.\ 0\notin\llbracket M_j\rrbracket$. 
The form of $D_i^{M_i}$ and the definition of $N_E$ imply that:
\[
\forall j.\ 1\leq j\leq n.\ \forall x\in \{0\}\cup \llbracket M_j\rrbracket.\ (a_j,x)\in N_E.
\]
The form of $D_i^{M_i}$ and the definition of $P_E$ imply that $\{a_1,\dots, a_n\}\in P_E$.
Since $w$ is consistent with $P_E$ and $N_E$, we infer that $\exists j.\ 1\leq j\leq n.\ a_j\in w$, and, moreover, $w(a_j)\in\llbracket M_j\rrbracket$.

The form of $D_i^{M_i}$ and the definition of $C_E$ imply that $\forall j,l\in\{1,\ldots,n\} .\ (j\neq l \Rightarrow (a_j, a_l)\in C_E)$, which implies that $\forall j,l\in\{1,\ldots,n\} .\ ((j\neq l\wedge a_j\in w) \Rightarrow a_l\notin w)$.

From the last two relations we obtain that: 
\[\exists j.\ 1\leq j\leq n.\ (w_i(a_j)\in \llbracket M_j\rrbracket\wedge \forall l.\ 1\leq l\leq n.\ (l\neq j\Rightarrow a_l \notin w_i)),
\]
in other words we have shown that $w_i\models D_i^{M_i}$.

\item $D_i^{M_i} = (a_1^{M_1}\mid\ldots\mid a_n^{M_n})\wedge \forall j.\ 1\leq j\leq n.\ 0\in\llbracket M_j\rrbracket$.
The reasoning is similar to the previous case, the only difference is that now $\{a_1,\dots, a_n\}\notin P_E$, so we obtain $w_i\models D_i^{M_i}$ even if none of the $a_j$ is present in $w_i$.
\end{enumerate}
From the three cases presented above we conclude that $w\models(C_E, N_E, P_E)\Rightarrow w\models E$.

For the \emph{only if} part, consider a normalized disjunctive multiplicity expression $E=D_1^{M_1}\shuffle\ldots\shuffle D_m^{M_m}$ and an unordered word $w$ such that $w\models E$.
This is equivalent to:
\[
\exists w_1,\dots,w_m.\ w=w_1\uplus\dots\uplus w_m \wedge \forall i.\ 1\leq i\leq m.\ w_i\models D_i^{M_i}.
\] 
We prove that $E\models(C_E,N_E,P_E)$ while reasoning on the three cases for $D_i^{M_i}$, for every $i$ such that $1\leq i\leq m$:
\begin{enumerate}
\item $D_i^{M_i} = (a_1\mid\ldots\mid a_n)^+$. In this case $w_i\models D_i^{M_i}$ implies that $\exists j.\ 1\leq j\leq n.\ a_j\in w_i$, so $\{a_1,\dots, a_n\}\in P_E$ is satisfied.
There are no conflicting pairs of symbols in $\{a_1,\dots, a_n\}$.
Since $\forall j.\ 1\leq j\leq n.\ \forall x\in\mathbb N_0.\ (a_j,x)\in N_E$, we obtain that in $w_i$ all the symbols have numbers of occurrences consistent with $N_E$.

\item $D_i^{M_i} = (a_1^{M_1}\mid\ldots\mid a_n^{M_n})\wedge \forall j.\ 1\leq j\leq n.\ 0\notin\llbracket M_j\rrbracket$. In this case $w_i\models D_i^{M_i}$ implies that:
\[
\exists j.\ 1\leq j\leq n.\ (w_i(a_j)\in \llbracket M_j\rrbracket\wedge \forall l.\ 1\leq l\leq n.\ (l\neq j\Rightarrow a_l \notin w_i)),
\]
which implies that $\{a_1,\dots, a_n\}\in P_E$ is satisfied.

The conflicting pairs of symbols are also satisfied, more precisely we know from the form of $D_i^{M_i}$ and the definition of $C_E$ that $\forall j,l\in\{1,\ldots,n\} .\ (j\neq l \Rightarrow (a_j, a_l)\in C_E)$. Moreover, $w_i\models D_i^{M_i}$ implies that $\forall j,l\in\{1,\ldots,n\} .\ (j\neq l\wedge a_j\in w_i \Rightarrow a_l\notin w_i)$, so there are no conflicts in $w_i$.

From the form of $D_i^{M_i}$ and the definition of $N_E$, we know that:
\[
\forall j.\ 1\leq j\leq n.\ \forall x\in \{0\}\cup \llbracket M_j\rrbracket.\ (a_j,x)\in N_E.
\]
We infer that $w_i$ is consistent with $N_E$ for the present symbol (since $w_i\models D_i^{M_i}$) and also for the symbols which are not present (since $0$ belongs to their extended cardinality map).

\item $D_i^{M_i} = (a_1^{M_1}\mid\ldots\mid a_n^{M_n})\wedge \forall j.\ 1\leq j\leq n.\ 0\in\llbracket M_j\rrbracket$.
In this case the reasoning for $C_E$ and $N_E$ is similar to the previous case.
The difference is that now $P_E$ is less restrictive, since $\{a_1,\dots,a_n\}\notin P_E$.
\end{enumerate}
From the three cases presented above we conclude that $w\models E\Rightarrow w\models(C_E, N_E, P_E)$
\qed\end{proof}
We also characterize the inclusion of two languages given by the characterizing triples:
\begin{lemma}\label{lemma:incl-tuple}
Given two disjunctive multiplicity expressions $E_1$ and $E_2$:
$(C_{E_1}\subseteq C_{E_2}\wedge N_{E_2}\subseteq N_{E_1}\wedge P_{E_1}\subseteq P_{E_2})$ iff $(\forall w.\ w\models (C_{E_2}, N_{E_2}, P_{E_2})\Rightarrow w\models (C_{E_1}, N_{E_1}, P_{E_1}))$.
\end{lemma}
\begin{proof}\normalfont
For the \emph{if} part, we prove by contraposition:
\begin{itemize}
\item $C_{E_1}\varnot\subseteq C_{E_2} \Rightarrow \exists (a_1,a_2)\in C_{E_1}.\ (a_1,a_2)\notin C_{E_2}\Rightarrow \exists (a_1,a_2)\in\Sigma\times\Sigma.\  (\varnot\exists w\in L(E_1).\ a_1\in w \wedge a_2\in w)\wedge(\exists w'\in L(E_2).\ a_1\in w'\wedge a_2\in w')
\Rightarrow (\exists w'.\ w'\models C_{E_2}\wedge w'\varnot\models C_{E_1})$.
\item $N_{E_2}\varnot\subseteq N_{E_1}\Rightarrow \exists a\in\Sigma.\ \exists w\in L(E_2).\ \varnot\exists w'\in  L(E_1).\ w'(a) = w(a)\Rightarrow (\exists w.\ w\models N_{E_2}\wedge w\varnot\models N_{E_1})$.
\item $P_{E_1}\varnot\subseteq P_{E_2}
\Rightarrow \exists X\subseteq \Sigma.\ (\forall w\in L(E_1).\ \exists a\in X.\ a\in w ) \wedge(\exists w'\in L(E_2).\ \forall a\in X.\ a \notin w')\Rightarrow (\exists w'.\ w'\models L(E_2)\wedge w'\varnot\models L(E_1))$. Using the previous Lemma, we infer that $(\exists w'.\ w'\models P_{E_2}\wedge w'\varnot\models P_{E_1}).$
\end{itemize}
For the \emph{only if} part, we take an unordered word $w$ such that $w\models (C_{E_2}, N_{E_2}, P_{E_2})$ and we want to prove that $w\models (C_{E_1}, N_{E_1}, P_{E_1})$, assuming that $C_{E_1}\subseteq C_{E_2}$, $N_{E_2}\subseteq N_{E_1}$, and $P_{E_1}\subseteq P_{E_2}$.

By definition, $w \models N_{E_2}$ implies that $\forall a \in \Sigma$, $(a,w(a)) \in N_{E_2}$. By hypothesis, $N_{E_2}\subseteq N_{E_1}$, therefore $\forall a \in \Sigma. (a,w(a)) \in N_{E_1}$, which by definition gives $w \models N_{E_1}$.

By definition, $w \models C_{E_2}$ implies that for all $(a, b) \in C_{E_2}$, (i) $(a\in w \Rightarrow b \notin w) \wedge (b\in w \Rightarrow a\notin w)$. By hypothesis, $C_{E_1}\subseteq C_{E_2}$, therefore (i) holds also for all $(a,b) \in C_{E_1}$, which by definition gives $w \models C_{E_1}$.

By definition, $w \models P_{E_2}$ implies that for all $X \in P_{E_2}$, (ii) $\exists a \in X$ s.t. $a \in w$. By hypothesis, $P_{E_1}\subseteq P_{E_2}$, therefore (ii) also holds for all $X \in P_{E_1}$, which by definition gives $w \models P_{E_1}$.
\qed\end{proof}
A consequence of Lemmas \ref{lemma:cons-tuple} and \ref{lemma:incl-tuple} is that the characterizing triples allow us to capture the containment of
disjunctive multiplicity expressions:
\begin{lemma}\label{lemma:characteristic-sets}
  Given two disjunctive multiplicity expressions $E_1$ and $E_2$,
  ${L}(E_2)\subseteq {L}(E_1)$ iff $C_{E_1}\subseteq C_{E_2}$,  
  $N_{E_2}\subseteq N_{E_1}$, and $P_{E_1}\subseteq P_{E_2}$.
\end{lemma}
The above lemma shows that two equivalent disjunctive multiplicity expressions yield the same triples and hence the triple $(C_E,N_E,P_E)$ can be viewed as a normal form for the languages definable by a DMS. Formally:
\begin{corollary}
  Given two disjunctive multiplicity expressions $E_1,E_2$, it holds that $L(E_1) = L(E_2)$ iff $C_{E_1} = C_{E_2}$, $N_{E_1} = N_{E_2}$, and $P_{E_1} = P_{E_2}$.
\end{corollary}

\subsubsection{Complexity results}
From Lemma~\ref{lemma:characteristic-sets} we know that the containment of two disjunctive multiplicity expressions is equivalent to the containment of their characterizing triples.
Next, we show that we can decide it in PTIME by using the compact representation of the characterizing triples:
\begin{lemma}\label{lemma:triple:ptime}
Given two disjunctive multiplicity expressions $E_1$ and $E_2$, deciding whether $L(E_2)\subseteq L(E_1)$ is in PTIME.
\end{lemma}
\begin{proof}\normalfont
From Lemma~\ref{lemma:characteristic-sets} we know that, given two disjunctive multiplicity expressions $E_1$ and $E_2$,
  ${L}(E_2)\subseteq {L}(E_1)$ iff $C_{E_1}\subseteq C_{E_2}$,  
  $N_{E_2}\subseteq N_{E_1}$, and $P_{E_1}\subseteq P_{E_2}$.
Note that testing $N_{E_2}\subseteq N_{E_1}$ is equivalent to testing whether $\forall a \in\Sigma.\ N_{E_2}^*(a)\subseteq N_{E_1}^*(a)$, which is in PTIME since it reduces to manipulating multiplicities.
Moreover, note that testing $P_{E_1}\subseteq P_{E_2}$ is equivalent to testing whether $\forall X\in P_{E_1}^*.\ \exists Y\in P_{E_2}^*.\ Y\subseteq X$, which is in PTIME since it reduces to testing the inclusion of a polynomial number of polynomial sets.
On the other hand, we can decide $C_{E_2}\subseteq C_{E_1}$ in PTIME without using the compact representation because each of these sets has a number of elements quadratic in $|\Sigma|$, and can be easily computed in $O(|\Sigma|^2)$.
\qed\end{proof}
Furthermore, testing the containment of two DMS reduces to testing, for each symbol in the alphabet, the containment of the associated disjunctive multiplicity expressions.
This problem is in PTIME (from Lemma~\ref{lemma:triple:ptime}).
Hence, we obtain:
\begin{theorem}\label{cnt-ptime}
$\CNT_{\dms}$ is in PTIME.
\end{theorem}
Next, we present the complexity results for satisfiability and
membership, and a \emph{streaming} algorithm for solving the membership.
The problem of validating a XML document with bounded memory was addressed in~\cite{SeSi07,SeVi02} and their conclusion is that constant memory validations can be performed only for some DTDs.
We propose a \emph{streaming algorithm} which processes an XML document in a single pass, using memory which depends on the height of the tree and not on its size.
%
%
For a tree $t$, $\mathit{height}(t)$ is the height of $t$
defined in the usual way. We employ the standard RAM model and assume
that subsequent natural numbers are used as labels in $\Sigma$, startig with $1$. 
\begin{proposition}\label{sat-memb-ptime} 
  Checking satisfiability of a DMS $S$ can be done in time
  $O(|\Sigma|^2)$. There exists a streaming algorithm that
  checks membership of a tree $t$ in a DMS $S$ in time
  $O(|\Sigma|\times|t| + |\Sigma|^2)$ and using space
  $O(\mathit{height}(t) \times |\Sigma| + |\Sigma|^2)$.
\end{proposition}
The algorithm first checks satisfiability of the schema, by performing a preprocessing in time $O(|\Sigma|^2)$, and then a simple process based on dynamic programming. If the schema is not satisfiable, the algorithm rejects the tree w/o reading anything on the stream.
Then the algorithm checks whether the schema is universal. A schema $S$ is \emph{universal} if the $L(R_S(root_S))$ is the set of all unordered words over $\Sigma$. This can be performed in time $O(|\Sigma|^2)$. 
If we assume that $\Sigma=\{a_1,\ldots,a_n\}$, a simple algorithm has to check whether each normalized disjunctive multiplicity expression from the rules of the schema has the form $a_1^*\shuffle\ldots\shuffle a_n^*$.

For checking in streaming the membership of a tree $t$ to the
language of a DMS $S$, the input tree $t$ is given in XML format. The algorithm
works for any arbitrary ordering of sibling nodes.  If the schema is universal, then the algorithm only reads the opening tag of the root of the tree. The tree is accepted if the label of the root is $root_S$, and rejected otherwise. Otherwise,
given a DMS $S$, in a preprocessing stage the algorithm constructs compact representations of the characterizing triples of the expressions used by $S$.
Remark that, as DMS forbids repetition of symbols, the size of the representation of any expression is linear in $|\Sigma|$.
Therefore, encoding the schema requires $O(|\Sigma|^2)$ space.
For each symbol $a\in\Sigma$, we encode its corresponding rule using three \emph{global dictionaries}, that we define as functions:
\begin{itemize}
\item $\cardinality_a:\Sigma\rightarrow\{0,1,?,+,*\}$ which represents the extended cardinality map of the disjunctive multiplicity expression $R_S(a)$.
\item $\conflict_a:\Sigma \rightarrow \{0,1,\ldots,|\Sigma|\}$ which encodes the conflicts from $R_S(a)$ and has the following properties:
\begin{itemize}
\item For any disjunction of the form $(a_1^{M_1}\mid\ldots\mid a_n^{M_n}) \textrm{ from } R_S(a).\ \conflict_a(a_1) = \ldots =\conflict_a(a_n) \wedge \forall a'\in \Sigma.\ a'\notin\{a_1,\ldots,a_n\}.\ \conflict_a(a')\neq \conflict_a(a_1),$
\item $\forall a'\in\Sigma.\ \forall X\in C_{R_S(a)}^*.\ a'\notin X.\ \conflict_a(a') = 0$.
\end{itemize}

Let $\mathcal C_a = \{x\in \mathbb \{0,1,\ldots,|\Sigma|\}\mid \exists a'\in\Sigma.\ \conflict_a(a')=x\}$.
\item $\required_a:\Sigma \rightarrow \{0,1,\ldots,|\Sigma|\}$ which encodes the sets of required symbols from $R_S(a)$ and has the following properties:
\begin{itemize}

\item For any disjunction of the form $(a_1\mid\ldots\mid a_n)^+$ or $(a_1^{M_1}\mid\ldots\mid a_n^{M_n}).\ 0\notin\llbracket M_1 \rrbracket$ from $R_S(a).\ \required_a(a_1)=\ldots=\required_a(a_n)\wedge \forall a'\in \Sigma.\ a'\notin\{a_1,\ldots,a_n\}.\ \required_a(a')\neq \required_a(a_1),$
\item $\forall a'\in\Sigma.\ \forall X\in P_{R_S(a)}^*.\ a'\notin X.\ \required_a(a') = 0$.
\end{itemize}

Let $\mathcal P_a = \{x\in \mathbb \{0,1,\ldots,|\Sigma|\}\mid \exists a'\in\Sigma.\ \required_a(a')=x\}$.
\end{itemize}
For example, assume the rule
$r \rightarrow (a\mid b)^+\shuffle (c^?\mid d^*\mid e^*)\shuffle f^+\shuffle g^?\shuffle (i\mid j^+)$
over the alphabet $\Sigma=\{a,b,c,d,e,f,g,h,i,j\}$. 
A possible encoding is the following:
\begin{gather*}
\cardinality_r(a) = \cardinality_r(b) = \cardinality_r(d) =  \cardinality_r(e) = \cardinality_r(j) =*,\\
\cardinality_r(c) = \cardinality_r(g) = \cardinality_r(i) = ?, \\
\cardinality_r(f) = +, \qquad \cardinality_r(h) = 0,\\
\conflict_r(a) = \conflict_r(b) = \conflict_r(f) =\conflict_r(g) =\conflict_r(h) = 0,\\
\conflict_r(c) = \conflict_r(d) =\conflict_r(e) = 1,\\
\conflict_r(i) = \conflict_r(j) =2,\\
\required_r(c) = \required_r(d) = \required_r(e) = \required_r(g) =\required_r(h) = 0,\\
\required_r(a) = \required_r(b) = 1, \qquad \required_r(f) = 2,\\
\required_r(i) = \required_r(j) = 3.
\end{gather*}
During the execution, the algorithm maintains a stack whose height is the depth of the currently visited node. 
The bound on space required for stack operations is $O(\mathit{height}(t)\times|\Sigma|)$.
We describe the \emph{local variables} for each node $n\in N_t$, more precisely three dictionaries (with size linear in $|\Sigma|$) that we define as functions:
\begin{itemize}
\item $\mcount :\Sigma\rightarrow \{0,1,2\}$ (initial value = $0$),
\item $\pconflict:\mathcal C_{\lab_t(n)}~\backslash~\{0\}\rightarrow \Sigma\cup\{0\}$ (initial value = $0$),
\item $\prequired:\mathcal P_{\lab_t(n)}~\backslash~\{0\}\rightarrow \{0, 1\}$ (initial value = $0$).
\end{itemize}
Next, we present Algorithms~\ref{open} and \ref{close}, which are executed when we encounter an opening or a closing tag, respectively.
The streaming algorithm rejects a tree as soon as the opening tag is read for nodes that violate either some conflicting pair (Algorithm~\ref{open}, lines 8-9) or the allowed cardinality (Algorithm~\ref{open}, lines 4-7).
The algorithm also rejects a tree if at the closing tag of a node, there are children symbols required by the corresponding rule of the node's label and not present in its children list (Algorithm~\ref{close}, lines 1-2).
Unless the schema is universal (i.e., accepts any tree), the acceptance of a tree can be decided only after the closing tag of the root.
\begin{algorithm}
\caption{\label{open}Procedure to execute when we are in a node $n\in N_t$ and we encounter an \emph{open tag} of a node $n_a$ labeled by $a$.}
\ALGORITHM $\open(n_a)$\\
\INPUT: Open tag of a node $n_a\in N_t$ labeled by $a$\\
\OUTPUT: Reject the tree or update the local variables\\
\LN push on the stack the local variables for $n_a$\\
\LN \IF $\mcount(a)\neq 2$ \THEN \\
\LN \TAB $\mcount(a)\colonequals\mcount(a)+1$ \\
\LN \IF $\mcount(a) = 2$ \AND $\cardinality_{\lab_t(n)}(a)\notin\{+,*\}$ \THEN\\ 
\LN \TAB \REJECT\\
\LN \IF $\mcount(a) = 1$ \AND $\cardinality_{\lab_t(n)}(a)=0$ \THEN\\
\LN \TAB \REJECT\\
\LN \IF $\conflict_{\lab_t(n)}(a)\neq 0$ \AND $\pconflict(\conflict_{\lab_t(n)}(a))\notin\{0,a\}$ \THEN\\ 
\LN \TAB \REJECT\\
\LN \IF $\conflict_{\lab_t(n)}(a)\neq 0$ \THEN\\ 
\LN \TAB $\pconflict(\conflict_{\lab_t(n)}(a)) \colonequals a$\\
\LN \IF $\required_{\lab_t(n)}(a)\neq 0$ \THEN\\ 
\LN \TAB $\prequired(\required_{\lab_t(n)}(a)) \colonequals 1$\\
\end{algorithm}
\begin{algorithm}
\caption{\label{close}Procedure to execute when we encounter the \emph{close tag} of a node $n\in N_t$.}
\ALGORITHM $\close(n)$\\
\INPUT: Close tag of a node $n\in N_t$\\
\OUTPUT: Accept or reject the tree, or continue\\
\LN \IF $\exists p\in\mathcal P_{\lab_t(n)}~\backslash~\{0\}.\ \prequired(p) = 0$ \THEN\\
\LN\TAB \REJECT\\
\LN pop the local variables for $n$ from the stack\\
\LN\IF $n=\root_t$ \THEN\\
\LN\TAB\ACCEPT
\end{algorithm}

A streaming algorithm is called {\em earliest} if it produces its result at the earliest point. More precisely, consider the algorithm processing an XML stream of tree $t$ for checking membership of $t$ to a schema $S$. At each position of the stream (i.e. each opening or closing tag), the algorithm has seen a part of the tree $t$, and another part of $t$ remains unknown at that position. Let $p$ be some position of the stream. If the tree $t$ would be accepted (resp. rejected) whatever the part of $t$ unknown at position $p$, then an earliest streaming algorithm has to accept (resp. reject) the tree at position $p$. For example, if the language of the schema is universal, then an earliest algorithm would accept or reject the tree as soon as the opening tag of the root is read.
It can be shown that the algorithm presented here is earliest.

\noindent We continue with complexity results that follow from known facts.
Query satisfiability for DTDs is known to be NP-complete~\cite{BeFaGe08} and we adapt the result for DMS:
\begin{proposition}\label{satq-nphard}
$\SAT_{\dms,\Twig}$ is NP-complete.
\end{proposition}
\begin{proof}\normalfont[sketch]
Proposition 4.2.1 from~\cite{BeFaGe08} implies that satisfiability of twig queries in the presence of DTDs is NP-hard.
We adapt the proof and we obtain the following reduction from $\sat$ to $\SAT_{\dms, \Twig}$: we take a $\CNF$ formula $\varphi=\bigwedge_{i=1}^nC_i$ over the variables $x_1,\ldots,x_m$, where each $C_i$ is a disjunction of 3 literals. 
Consider $\Sigma=\{r, t_1,f_1,\ldots,t_m,f_m,C_1,\ldots,C_n\}$ and the corresponding tuple $(S, q)$:
\begin{itemize}
\item The schema $S$ having the root label $r$ and the rules:
\begin{itemize}
\item $r \rightarrow (t_1 \mid f_1) \shuffle \dots\shuffle(t_m\mid f_m)$
\item $t_j \rightarrow C_{j_1}\shuffle\ldots\shuffle C_{j_k},\ 1\leq j\leq m.~x_j$ appears in $C_{j_i}$
\item $f_i  \rightarrow C_{j_1}\shuffle\ldots\shuffle C_{j_k},\ 1\leq j\leq m.~\neg x_j$ appears in $C_{j_i}$
\end{itemize}
\item The query $q = r[\dblslash C_1]\dots[\dblslash C_n]$
\end{itemize}
For example, for the $\CNF$ formula over the variables $x_1,\ldots, x_4$: $\varphi_0=(x_1\vee \neg x_2 \vee x_3)\wedge(\neg x_1\vee x_3\vee \neg x_4)$ we have the schema $S$ containing the rules:
\begin{flalign*}
r \rightarrow (t_1\mid f_1) \shuffle (t_2\mid f_2)& \shuffle (t_3\mid f_3) \shuffle (t_4\mid f_4) \\
t_1 \rightarrow C_1 ~~~& t_3\rightarrow C_1\shuffle C_2 \\
f_1\rightarrow C_2 ~~~& f_3 \rightarrow \epsilon\\
t_2\rightarrow \epsilon ~~~& t_4 \rightarrow \epsilon \\
f_2\rightarrow C_1 ~~~& f_4 \rightarrow C_2
\end{flalign*}
and the query:
\[
q = /r[\dblslash C_1][\dblslash C_2]
\]
The formula $\varphi$ is satisfiable iff $(S, q)\in \SAT_{\dms, \Twig}$.
The described reduction works in polynomial time in the size of the input formula $\varphi$.
Moreover, Theorem 4.4 from~\cite{BeFaGe08} implies that satisfiability of twig queries in the presence of DTDs is in NP, which yields the NP upper bound for $\SAT_{\dms, \Twig}$.\qed
\end{proof}
The complexity results for query implication and query containment in
the presence of DMS follow from the EXPTIME-completeness proof from
\cite{NeSc06} for twig query containment in the presence of DTDs.
\begin{proposition}\label{exptime-c}
$\IMPL_{\dms,\Twig}$ and $\CNT_{\dms,\Twig}$ are EXPTIME-complete.
\end{proposition}
\begin{proof}\normalfont[sketch]
Theorem 4.4 from~\cite{NeSc06} implies that twig query containment in the presence of DTDs is in EXPTIME.
This implies that the problems $\IMPL_{\DTD,\Twig}$, $\IMPL_{\dms,\Twig}$, and $\CNT_{\dms,\Twig}$ are also in EXPTIME.
The EXPTIME-hardness proof of twig containment in the presence of DTDs (Theorem 4.5 from~\cite{NeSc06}) has been done using a reduction from \emph{Two-player corridor tiling} problem and a technique introduced in \cite{MiSu04}.
In the proof from \cite{NeSc06}, when testing inclusion $p\subseteq_Sq$, $p$ is chosen such that it satisfies any tree in $S$, hence $\IMPL_{\DTD,\Twig}$ is also EXPTIME-complete.
Furthermore, Lemma 3 in \cite{MiSu04} can be adapted to twig queries and DMS: for any $S\in\dms$ and twig queries $q_0,q_1,\ldots,q_m$ there exists $S'\in\dms$ and twig queries $q$ and $q'$ such that:
\[
q_0 \subseteq_S q_1\cup\ldots\cup q_m \iff q \subseteq_{S'} q'.
\]
Because the DTD in \cite{NeSc06} can be captured with DMS, from the last two statements we conclude that $\IMPL_{\dms,\Twig}$ and $\CNT_{\dms,\Twig}$ are also EXPTIME-complete. 
\qed
\end{proof}

\subsection{Disjunction-free multiplicity schema}\label{subsec:ms}
In this subsection we present the static analysis for MS.  Although
query satisfiability and query implication are intractable for DMS,
these problems become tractable for MS because they can be reduced to
testing embedding of queries in some dependency graphs that we define
in the sequel. 
We first present some of the technical tools which help us to reason about the disjunction-free multiplicity schemas.
Next, we use these tools to prove our results.
Recall that MS use expressions of the form
$a_1^{M_1}\shuffle\ldots\shuffle a_n^{M_n}$.
\subsubsection{Dependency graphs}
\begin{definition}
  Given an MS $S=(\root_S, R_S)$, the \emph{dependency graph} of $S$
  is a directed rooted graph $G_S = (\Sigma,\root_S, E_S)$ with the
  node set $\Sigma$, where $\root_S$ is the distinguished root node,
  and $(a,b)\in E_S$ if $R_S(a) =\ldots\shuffle b^M\shuffle\ldots$ and
  $M\in\{*,+,?,1\}$. Furthermore, the edge $(a,b)$ is called
  \emph{nullable} if $0\in\llbracket M\rrbracket$ (i.e., $M$ is $*$ or
  $?$), otherwise $(a,b)$ is called \emph{non-nullable} (i.e., $M$ is
  $+$ or $1$).  The \emph{universal dependency graph} of an MS $S$ is
  the subgraph $G_{S}^{\mathrm u}$ containing only the non-nullable
  edges.
\end{definition}
In Figure~\ref{figgraphs} we present the dependency graphs
for the schema $S_5$ containing the rules $r\rightarrow a^+\shuffle b^*,~ a\rightarrow b^?,~b\rightarrow\epsilon$.
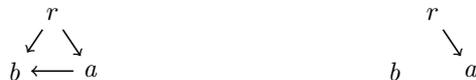
\begin{figure}[htb]
  \centering
  \begin{tikzpicture}[yscale=0.75]
    \begin{scope}
      \node at (0,0) (m99) {$r$};
      \node at (0.5,-1) (m0) {$\mathit{a}$};
      \node at (-0.5,-1) (m20) {$\mathit{b}$};
      \draw[->,semithick] (m99) -- (m0);			
      \draw[->,semithick] (m99) -- (m20);
      \draw[->,semithick] (m0) -- (m20);
    \end{scope}
    \begin{scope}[xshift=5cm]
      \node at (0,0) (m99) {$r$};
      \node at (0.5,-1) (m0) {$\mathit{a}$};
      \node at (-0.5,-1) (m20) {$\mathit{b}$};
      \draw[->,semithick] (m99) -- (m0);	
    \end{scope}
  \end{tikzpicture}
  \caption{\label{figgraphs}Dependency graph $G_{S_5}$ and universal dependency graph
    $G_{S_5}^{\mathrm u}$ for schema $S_5$.}
\end{figure}

\noindent An MS $S$ is \emph{pruned} if $G_S^{\mathrm u}$ is acyclic.  We
observe that any MS has an equivalent pruned version which can be
constructed in PTIME by removing the rules for the labels from which a
cycle can be reached in the universal dependency graph. Note that a
schema is satisfiable iff no cycle can be reached from its root in the
universal dependency graph.  From now on, we assume w.l.o.g.\ that all
the MS that we manipulate are pruned.

We generalize the notion of embedding as a mapping of the nodes of a
query $q$ to the nodes of a rooted graph $G=(\Sigma, \root,E)$,
which can be either a dependency graph or a universal dependency graph.
Formally, an \emph{embedding} of $q$ in $G$ is 
a function $\lambda : N_q \rightarrow \Sigma$ such that:
\begin{enumerate}
\itemsep0pt
\item[$1$.] $\lambda(\root_q)=\root$,
\item[$2$.] for every $(n,n')\in \child_q$,
  $(\lambda(n),\lambda(n'))\in E$,
\item[$3$.] for every $(n,n')\in \desc_q$,
  $(\lambda(n),\lambda(n'))\in E^+$ (the transitive closure of $E$),
\item[$4$.] for every $n\in N_q$, $\lab_q(n) = \wc$ or $\lab_q(n) = \lambda(n)$.
\end{enumerate}
If there exists an embedding from $q$ to $G$, we write $G\preccurlyeq q$.

\subsubsection{Graph simulation}
A \emph{simulation} of a rooted graph (either dependency graph or universal dependency graph) $G=(\Sigma,\root,E)$ in a tree $t$ is a relation $R\subseteq \Sigma\times N_t$ such that:
\begin{enumerate}
\item $(\root, \root_t)\in R$
\item for every $(a, n)\in R,\ (a, a')\in E$, there exists $n'\in N_t$
  such that $(n,n')\in\child_t$ and $(a',n')\in R$
\item for every $(a, n)\in R.\ \lab_t(n) = a$
\end{enumerate}
Note that $R$ is a total relation for the nodes of the graph reachable from the root i.e., for every $a\in \Sigma$ reachable from $\root$ in $G$, there exists a node $n\in N_t$ such that $(a,n)\in R$. 
If there exists a simulation from $G$ to $t$, we write $t\preccurlyeq G$.
The \emph{language of a graph} is $ L(G) = \{t\in\Tree\mid t\preccurlyeq G\}$.

A rooted graph $G_1=(\Sigma,\root,E_1)$ is a \emph{subgraph} of another rooted graph $G_2=(\Sigma,\root,E_2)$ if $E_1\subseteq E_2$.
For a rooted graph $G=(\Sigma,\root,E)$, we define the partial order $\leq_G$ on the subgraphs of $G$: given $G_1$ and $G_2$ two subgraphs of $G$, $G_1\leq_G G_2$ if $G_1$ is a subgraph of $G_2$.
Note that the relation $\leq_G$ is reflexive, antisymmetric, and transitive, thus being an order relation.
Moreover, it is well-founded and it has a minimal element, that we denote $G_0$ for a graph $G$. 
Let $G_0 = (\Sigma,\root, \emptyset)$ and indeed, for any $G'$ subgraph of $G$ we have $G_0 \leq_G G'$.
In the sequel, we assume w.l.o.g.\ that all the subgraphs that we use in our proofs have the property that every edge can be part of a path starting at the root.

\begin{lemma}\label{L1}
For any disjunction-free multiplicity schema $S$, its universal dependency graph can be simulated in any tree $t$ which belongs to the language of $S$:
\[
\forall S\in\ms.\ \forall t\in L(S).\ t\preccurlyeq G_S^{\mathrm u}
\]
\end{lemma}
\begin{proof}\normalfont
Consider an MS $S$ and its universal dependency graph $G_S^{\mathrm u}$. 
Let $t$ be a tree which belongs to $ L(S)$. 
We want to construct a witness relation $R\subseteq \Sigma \times N_t$ for $t\preccurlyeq G_S^{\mathrm u}$ and the proof goes by induction on the structure of $G_S^{\mathrm u}$, using the well-founded order $\leq_{G_S^{\mathrm u}}$ defined above. 
Let $P(G)$ denote the statement $t\preccurlyeq G$.
Let $G$ be a subgraph of $G_S^{\mathrm u}$. 
The \emph{induction hypothesis} is that for all $G'\leq_{G_S^{\mathrm u}} G$ and $G'\neq G$, there exists a relation $R'$ witness of the simulation $t\preccurlyeq G'$ and we are going to construct $R$ that witnesses $t\preccurlyeq G$. 

For the \emph{base case}, we take the minimal element for the relation $\leq_{G_S^{\mathrm u}}$ let it $G_0 = (\Sigma,\root,\emptyset)$, then $P(G_0)$ holds for the relation $R_0=\{(\root, \root_t)\}$, so the subgraph containing no edge can be simulated in $t$.

For the \emph{induction case}, let $G$ a subgraph of $G_S^{\mathrm u}$. 
By the induction hypothesis, we know that $P(G')$ holds, for every $G'\leq_{G_S^{\mathrm u}} G$. Consider a subgraph $G'$ of $G$ such that $G$ contains exactly one additional edge w.r.t.\ $G'$, let the additional edge $(a,a')$ and $R'$ the witness relation for $t\preccurlyeq G'$. 
Because $G'\leq_{G_S^{\mathrm u}}G$ and $(a, a')$ is the only additional edge, we know that $R'$ already contains images for $a$ in $t$ i.e., there exists a node $n$ such that $(a, n)\in R'$. 
We construct the relation $R$ as the union of $R'$ with $\{(a', n') \mid \lab_t(n') = a'\wedge (\exists n.\ (n, n') \in \child_t\wedge (a,n)\in R')\}$. 
The set of tuples that we add is not empty because the edge $(a, a')$ belongs to the universal dependency graph $G_S^{\mathrm u}$, so for any node labeled by $a$ in the tree $t$ there exists a child of it labeled with $a'$. 
The construction ensures that $R$ satisfies all the conditions of the definition of a simulation, so $t\preccurlyeq G$, so $P(G)$ is true.

We have proved that $P(G_0)$ is true and that ($\forall G'.\ G'\leq_{G_S^{\mathrm u}} G \Rightarrow P(G'))  \Rightarrow P(G)$, so $P(G)$ is true for any $G$ subgraph of $G_S^{\mathrm u}$, so also for $G_S^{\mathrm u}$, hence $G_S^{\mathrm u}$ can be simulated into any tree $t$ which belongs to the language of $S$.
\qed
\end{proof}

\subsubsection{Graph unfolding}
A \emph{path in a rooted graph} (either dependency graph or universal dependency graph) $G =(\Sigma, \root, E)$ is a non-empty sequence of vertices starting at $\root$ such that for any two consecutive vertices in the sequence, there is a directed edge between them in $G$. 
By $\Paths(G)\subseteq\Sigma^+$ we denote the set of all the paths in $G$. 
The set of paths is finite only for graphs without cycles reachable from the root. 
For instance, the paths of the graph $G_1$ in Figure~\ref{fig:unfolding} are $\Paths(G_1) = \{r, ra, rb, rc, rbd, rcd, rbde, rcde\}$. 

Similarly, a \emph{path in a tree} $t$ is a non-empty sequence of
nodes starting at $\root_t$ such that any two consecutive nodes in the
sequence are in the relation $child_t$. 
By $\Paths(t)$ we denote the set of all the paths in $t$. 
Then, we define $\LabPaths(t)$ as the set of sequences of labels of nodes from all the paths in $t$. 
For instance, for the tree $t_1$ from Figure~\ref{fig:LabPaths} we have
$\Paths(t_1)=\{n_0,n_0n_1,n_0n_1n_2,n_0n_3,n_0n_3n_4\}$ and
$\LabPaths(t_1)=\{r,ra,rab\}$. Note that $\Paths(t)\subseteq N_t^+$,
$\LabPaths(t)\subseteq \Sigma^+$ and $|\LabPaths(t)|\leq|\Paths(t)|$.
The \emph{unfolding} of a rooted graph $G=(\Sigma,\root,E)$, denoted $u_G$, is a tree $u_G = (N_{u_G}, \root_{u_G}, \lab_{u_G}, \child_{u_G})$, such that:
\begin{itemize}
\item $N_{u_G} = \Paths(G)$,
\item $\root_{u_G}\in N_{u_G}$ is the root of $u_G$,
\item $(p, p.a)\in\child_{u_G}$, for all paths $p, p.a\in \Paths(G)$ (note that ``.'' stands for concatenation),
\item $\lab_{u_G}(\root_{u_G}) = \root$, and $\lab_{u_G}(p.a) = a$, for all the paths $p.a\in\Paths(G)$. 
\end{itemize}
The unfolding of a graph is finite only when the graph has no cycle reachable from the root, because otherwise $\Paths(G)$ is infinite, so $u_G$ is infinite. 
In the sequel we use the unfolding for graphs without any cycle reachable from the root and in this case the unfolding is the \emph{smallest} tree ${u_G}$ (w.r.t.\ the number of nodes) having $\LabPaths({u_G}) = \Paths(G)$.
The idea of the unfolding is to transform the rooted graph $G$ into a tree having the $\child$ relation instead of directed edges. 
There are nodes duplicated in order to avoid nodes with more than one incoming edge.
For instance, in Figure~\ref{fig:unfolding} we take the graph $G_1$ and construct its unfolding $u_{G_1}$. We remark that the size of the unfolding may be exponential in the size of the graph, for example for the graph $G_2$ from Figure \ref{fig:exponential unfolding}.
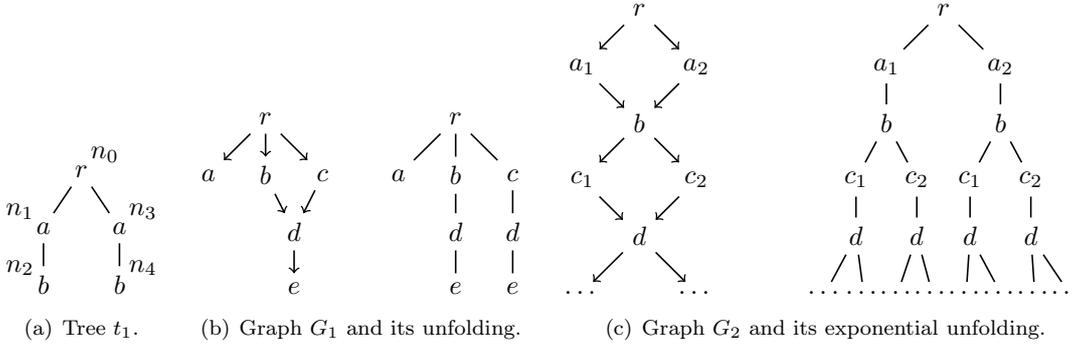
\begin{figure}[htb]
	\centering
	\subfigure[Tree $t_1$.]{\label{fig:LabPaths}
  \begin{tikzpicture}[yscale=0.75]
    \begin{scope}
      \node at (0,1) (m99) {$r$};
      \node at (0.5,0) (m0) {$a$};
      \node at (-0.5,0) (m1) {$a$};
      \node at (0.5,-1) (m2) {$b$};
      \node at (-0.5,-1) (m23) {$b$};
      \draw[-,semithick] (m99) -- (m0);	
      \draw[-,semithick] (m99) -- (m1);		
      \draw[-,semithick] (m0) -- (m2);
      \draw[-,semithick] (m1) -- (m23);
      \path (m99) node[above right] {$n_0$};
      \path (m1) node[above left] {$n_1$};
      \path (m23) node[above left] {$n_2$};
      \path (m2) node[above right] {$n_4$};
      \path (m0) node[above right] {$n_3$};
    \end{scope}
  \end{tikzpicture}
	}
\subfigure[Graph $G_1$ and its unfolding.]{\label{fig:unfolding}
  \begin{tikzpicture}[yscale=0.75]
    \node at (0,0) (n0) {$r$};
    \node at (-0.75,-1) (n1) {$a$};
    \node at (0,-1) (n2) {$b$};
    \node at (0.75,-1) (n3) {$c$};
    \node at (0.375,-2) (n4) {$d$};
    \node at (0.375,-3) (n5) {$e$};
    \draw[->,semithick] (n0) -- (n1);
    \draw[->,semithick] (n0) -- (n2);
    \draw[->,semithick] (n0) -- (n3);
    \draw[->,semithick] (n2) -- (n4);
    \draw[->,semithick] (n3) -- (n4);
    \draw[->,semithick] (n4) -- (n5);
    \begin{scope}[xshift=2.5cm]
    \node at (0,0) (n0) {$r$};
    \node at (-0.75,-1) (n1) {$a$};
    \node at (0,-1) (n2) {$b$};
    \node at (0.75,-1) (n3) {$c$};
    \node at (0,-2) (n4) {$d$};
    \node at (0.75,-2) (n5) {$d$};
    \node at (0,-3) (n6) {$e$};
    \node at (0.75,-3) (n7) {$e$};
    \draw[-,semithick] (n0) -- (n1);
    \draw[-,semithick] (n0) -- (n2);
    \draw[-,semithick] (n0) -- (n3);
    \draw[-,semithick] (n2) -- (n4);
    \draw[-,semithick] (n3) -- (n5);
    \draw[-,semithick] (n4) -- (n6);
    \draw[-,semithick] (n5) -- (n7);
    \end{scope}
  \end{tikzpicture}
}
\subfigure[Graph $G_2$ and its exponential unfolding.]{  \label{fig:exponential unfolding}
  \begin{tikzpicture}[yscale=0.75]
    \node at (0,0) (n0) {$r$};
    \node at (-0.75,-1) (n1) {$a_1$};
    \node at (0.75,-1) (n2) {$a_2$};
    \node at (0,-2) (n3) {$b$};
    \node at (-0.75,-3) (n4) {$c_1$};
    \node at (0.75,-3) (n5) {$c_2$};
    \node at (0,-4) (n6) {$d$};
    \node at (-0.75,-5) (n7) {$\dots$};
    \node at (0.75,-5) (n8) {$\dots$};
    \draw[->,semithick] (n0) -- (n1);
    \draw[->,semithick] (n0) -- (n2);
    \draw[->,semithick] (n1) -- (n3);
    \draw[->,semithick] (n2) -- (n3);
    \draw[->,semithick] (n3) -- (n4);
    \draw[->,semithick] (n3) -- (n5);
    \draw[->,semithick] (n4) -- (n6);
    \draw[->,semithick] (n5) -- (n6);
    \draw[->,semithick] (n6) -- (n7);
    \draw[->,semithick] (n6) -- (n8);
    \begin{scope}[xshift=4cm]
    \node at (0,0) (n0) {$r$};
    \node at (-0.75,-1) (n1) {$a_1$};
    \node at (0.75,-1) (n2) {$a_2$};
    \node at (-0.75,-2) (n3) {$b$};
    \node at (0.75,-2) (n30) {$b$};
    \node at (-1.15,-3) (n4) {$c_1$};
    \node at (-0.35,-3) (n5) {$c_2$};
    \node at (1.15,-3) (n40) {$c_2$};
    \node at (0.35,-3) (n50) {$c_1$};
    \node at (-1.15,-4) (n6) {$d$};
    \node at (-0.35,-4) (n60) {$d$};
    \node at (0.35,-4) (n61) {$d$};
    \node at (1.15,-4) (n62) {$d$};
    \node at (-1.55,-5) (n7) {$\dots$};    
   \node at (-1.05,-5) (n70) {$\dots$};
    \node at (-0.6,-5) (n71) {$\dots$};
    \node at (-0.15,-5) (n72) {$\dots$};
    \node at (0.3,-5) (n73) {$\dots$};
    \node at (0.75,-5) (n74) {$\dots$};
    \node at (1.2,-5) (n75) {$\dots$};
    \node at (1.65,-5) (n76) {$\dots$};
    \draw[-,semithick] (n0) -- (n1);
    \draw[-,semithick] (n0) -- (n2);
    \draw[-,semithick] (n1) -- (n3);
    \draw[-,semithick] (n2) -- (n30);
    \draw[-,semithick] (n3) -- (n4);
    \draw[-,semithick] (n3) -- (n5);
    \draw[-,semithick] (n30) -- (n40);
    \draw[-,semithick] (n30) -- (n50);
    \draw[-,semithick] (n60) -- (n5);
    \draw[-,semithick] (n6) -- (n4);
    \draw[-,semithick] (n61) -- (n50);
    \draw[-,semithick] (n62) -- (n40);
    \draw[-,semithick] (n6) -- (n7);
    \draw[-,semithick] (n6) -- (n70);
    \draw[-,semithick] (n60) -- (n71);
    \draw[-,semithick] (n60) -- (n72);
    \draw[-,semithick] (n61) -- (n73);
    \draw[-,semithick] (n61) -- (n74);
    \draw[-,semithick] (n62) -- (n75);
    \draw[-,semithick] (n62) -- (n76);
    \end{scope}
  \end{tikzpicture}
}
\caption{A tree and two graphs with their corresponding unfoldings.}
\end{figure}

\subsubsection{Extending the definition of embedding}
If a query $q$ can be embedded in a tree $t$, we may write $t\preccurlyeq q$ instead of $t\models q$.
We also extend the definition of embedding from a query to a tree to the embedding from a tree to another tree i.e., given two trees $t$ and $t'$, we say that $t'$ can be embedded in $t$ (denoted $t\preccurlyeq t'$) if the query $(N_{t'},\root_{t'},\lab_{t'},\child_{t'},\emptyset)$ can be embedded in $t$.
Similarly, we can define the embedding from a tree to a rooted graph.
Note that two embeddings can be composed, for example:
\begin{itemize}
\item $\forall t,t'\in\Tree.\ \forall q\in\Twig.\ (t\preccurlyeq t'\wedge t'\preccurlyeq q\Rightarrow t\preccurlyeq q)$.
\item $\forall S\in\ms.\ \forall t\in\Tree.\ \forall q\in\Twig.\ (G_S^{(\mathrm u)}\preccurlyeq t\wedge t\preccurlyeq q\Rightarrow G_S^{(\mathrm u)}\preccurlyeq q)$.
\end{itemize}
\begin{lemma}\label{L2}
  A rooted graph (dependency graph or universal dependency graph) $G=(\Sigma,\root,E)$ can be simulated in a tree $t$
  iff its unfolding $u_G$ can be embedded in $t$.
\end{lemma}
\begin{proof}\normalfont
For the \emph{if} part, we know that $t\preccurlyeq u_G$ so there exists a function $\lambda : N_{u_G}\rightarrow N_t$ which witnesses the embedding of $u_G$ in $t$. We construct a relation $R\subseteq \Sigma\times N_t$ such that:
\begin{gather*}
R = \{(\root, \root_t)\} \cup \{(a,n)\mid\exists p \in N_{u_G}.\ p.a\in N_{u_G}  \wedge \lambda(p.a) = n\}
\end{gather*}
This construction ensures that for every $(a, n) \in R$ and for every $(a, a')\in E$, there exists $n'\in N_t$ such that $(n,n')\in child_t$ and $(a', n')\in R$ because the function $\lambda$ is a witness for $t\preccurlyeq u_G$ so the $\child$ relation is simply translated from $u_G$ to $G$. The construction of $R$ also guarantees that for every $(a, n)\in R$ we have $\lab_t(n) = a$ because $\lambda$ is the witness for $t\preccurlyeq u_G$ and $\lambda(p.a) = n$. Thus we obtain that $R$ satisfies all the conditions to be a simulation of $G$ in $t$.

For the \emph{only if} case, we take a relation $R$ which witnesses the simulation of $G$ in $t$.
We construct the function $\lambda : N_{u_G}\rightarrow N_t$, witness of $t\preccurlyeq u_G$, by recursion on the paths of $G$, because $\Paths(G) = N_{u_G}$.
First of all, $\lambda(\root_{u_G}) = \root_t$.
We assume that we have a recursive procedure which takes as input a path $p$, a label $a$, and the values of the function $\lambda$ computed before the procedure call, and it outputs $\lambda(p.a)$.
The invariant of the procedure is that while defining $\lambda$ for $p.a$, $\lambda$ satisfies the conditions from the definition of embedding for all the nodes $\root_{u_G}, \dots, p$ on the path to $p$. 
Furthermore, the values of $\lambda$ were obtained using the information given by $R$, so $\lambda(p) = n'$ iff $R(\lab_t(n'), n')$.
Let $\lambda(p) = n'$ and we construct $\lambda(p.a) = n$, where $R(a, n)$ and $\child_t(n', n)$.
There exists such a node $n$ because of the recursive construction of $\lambda$ using $R$ and the invariant $\lambda(p.a) = n$ iff $R(a, n)$ is true. 
The construction of $\lambda$ ensures that $\lambda$ is root-preserving, child-preserving and label-preserving, so it satisfies all the conditions to be an embedding from $u_G$ to $t$, so we have found a correct witness for $t\preccurlyeq u_G$.
\qed\end{proof}

\begin{lemma}\label{L3}
  A query $q$ can be embedded in a rooted graph (dependency graph or universal dependency graph) $G$ iff $q$ can be embedded
  in the unfolding tree of $G$.
\end{lemma}
\begin{proof}\normalfont
For the \emph{if} part, we know that $u_G\preccurlyeq q$, so there exists a function $\lambda : N_q\rightarrow N_{u_G}$ witness of this embedding. 
We construct a function $\lambda': N_q\rightarrow \Sigma$, such that $\lambda'(n) = \lab_{u_G}(\lambda(n))$ for each node $n$ from $N_q$. 
Since $\lambda$ is the witness of the embedding $u_G\preccurlyeq q$, the constructed $\lambda'$ satisfies all the conditions of the definition of an embedding from $q$ to $G$.

For the \emph{only if} part, we know that $G\preccurlyeq q$, so there exists a function $\lambda : N_q\rightarrow \Sigma$ witness of this embedding. 
We want to construct a function $\lambda': N_q\rightarrow N_{u_G}$ to prove $u_G\preccurlyeq q$. 
We construct $\lambda'$ by recursion on the tree structure of $q$. First of all, $\lambda'(\root_q) = \root_{u_G}$. 
Then, the \emph{recursion hypothesis} says that $G\preccurlyeq q'$ for any connected subtree $q'$ obtained from $q$ by deleting some edges, $u_G\preccurlyeq q'$, which is witnessed by the function $\lambda'$. 
Thus, for any node $n$ of $q$, $\lambda'(n) = p$, where $p\in N_{u_G}$ because $N_{u_G} = \Paths(G)$ so any node in the unfolding can be identified by a unique sequence of labels among the paths of $G$.
For the \emph{inductive case} consider that $q$ is obtained from $q'$ by adding one more edge, let it $(n, n')$. 
If it is a child edge and $\lambda'(n) =p$, we construct $\lambda'(n') = p.\lambda(n')$, which is a path in $G$ by the definition of the unfolding. 
Otherwise, if it is a descendant edge and $\lambda'(n) =p$, we construct $\lambda'(n') = p.p'.\lambda(n')$, where $p'$ is a randomly chosen path in $G$ from $\lambda(n)$ to $\lambda(n')$. 
We know by definition of $\lambda$ that such path exists.
The construction ensures that $u_G\preccurlyeq q$, for any $q$ satisfying the conditions of the recursion, so we can construct a function $\lambda'$ which is a correct witness for $u_G\preccurlyeq q$.
\qed\end{proof}

\subsubsection{Fuse and add operations}
In Figure \ref{fuseadd} we present the operations \emph{fuse} and \emph{add}. 
We say that $t\lhd_0 t'$ if $t'$ is obtained from $t$ by applying one of the operations from Figure \ref{fuseadd}.
The \emph{fuse} operation takes two siblings with the same label and creates only one node having below it the subtrees corresponding to each of the siblings.
The \emph{add} operation consists simply in adding a subtree at any place in the tree.
By $\unlhd$ we denote the transitive and reflexive closure of $\lhd_0$.

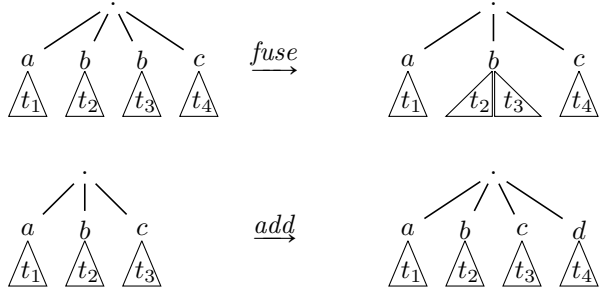
\begin{figure}[htb] 
  \centering
  \begin{tikzpicture}[yscale=0.75]
   \node at (0.375,0) (n0) {$.$};
   \node at (-0.75,-1) (n1) {$a$};
   \node at (0,-1) (n2) {$b$};
   \node at (0.75,-1) (n3) {$b$};
   \node at (1.5,-1) (n4) {$c$};
   \draw[-,semithick] (n0) -- (n1);
   \draw[-,semithick] (n0) -- (n2);
   \draw[-,semithick] (n0) -- (n3);
   \draw[-,semithick] (n0) -- (n4);
   \draw (-0.75,-1.2) -- (-1,-2) -- (-0.5,-2)-- (-0.75,-1.2);
   \node at (-0.7,-1.75) (t1) {$t_1$};
   \draw (0,-1.2) -- (-0.25,-2) -- (0.25,-2)-- (0,-1.2);
   \node at (0.05,-1.75) (t2) {$t_2$};
   \draw (0.75,-1.2) -- (0.5,-2) -- (1,-2)-- (0.75,-1.2);
   \node at (0.8,-1.75) (t3) {$t_3$};
   \draw (1.5,-1.2) -- (1.25,-2) -- (1.75,-2)-- (1.5,-1.2);
   \node at (1.55,-1.75) (t3) {$t_4$};
  \begin{scope}[xshift=2.5cm]
   \node at (0,-1) (n0) {$\underrightarrow{\mathit{fuse}}$};
  \end{scope}
  \begin{scope}[xshift=5cm]
   \node at (0.375,0) (n0) {$.$};
   \node at (-0.75,-1) (n1) {$a$};
   \node at (0.375,-1) (n2) {$b$};
   \node at (1.5,-1) (n4) {$c$};
   \draw[-,semithick] (n0) -- (n1);
   \draw[-,semithick] (n0) -- (n2);
   \draw[-,semithick] (n0) -- (n4);
   \draw (-0.75,-1.2) -- (-1,-2) -- (-0.5,-2)-- (-0.75,-1.2);
   \node at (-0.7,-1.75) (t1) {$t_1$};
   \draw (0.36,-1.2) -- (-0.25,-2) -- (0.36,-2)-- (0.36,-1.2);
   \node at (0.2,-1.75) (t2) {$t_2$};
   \draw (0.39,-1.2) -- (0.39,-2) -- (1,-2)-- (0.39,-1.2);
   \node at (0.65,-1.75) (t3) {$t_3$};
   \draw (1.5,-1.2) -- (1.25,-2) -- (1.75,-2)-- (1.5,-1.2);
   \node at (1.55,-1.75) (t3) {$t_4$};
  \end{scope}
  \begin{scope}[yshift=-3cm]
    \node at (0,0) (n0) {$.$};
    \node at (-0.75,-1) (n1) {$a$};
    \node at (0,-1) (n2) {$b$};
    \node at (0.75,-1) (n3) {$c$};
    \draw[-,semithick] (n0) -- (n1);
    \draw[-,semithick] (n0) -- (n2);
    \draw[-,semithick] (n0) -- (n3);
    \draw (-0.75,-1.2) -- (-1,-2) -- (-0.5,-2)-- (-0.75,-1.2);
    \node at (-0.7,-1.75) (t1) {$t_1$};
    \draw (0,-1.2) -- (-0.25,-2) -- (0.25,-2)-- (0,-1.2);
    \node at (0.05,-1.75) (t2) {$t_2$};
    \draw (0.75,-1.2) -- (0.5,-2) -- (1,-2)-- (0.75,-1.2);
    \node at (0.8,-1.75) (t3) {$t_3$};
  \end{scope}
  \begin{scope}[xshift=2.5cm,yshift=-3cm]
   \node at (0,-1) (n0) {$\underrightarrow{\mathit{add}}$};
  \end{scope}
  \begin{scope}[xshift=5cm,yshift=-3cm]
    \node at (0.375,0) (n0) {$.$};
    \node at (-0.75,-1) (n1) {$a$};
    \node at (0,-1) (n2) {$b$};
    \node at (0.75,-1) (n3) {$c$};
    \node at (1.5,-1) (n4) {$d$};
    \draw[-,semithick] (n0) -- (n1);
    \draw[-,semithick] (n0) -- (n2);
    \draw[-,semithick] (n0) -- (n3);
    \draw[-,semithick] (n0) -- (n4);
   \draw (-0.75,-1.2) -- (-1,-2) -- (-0.5,-2)-- (-0.75,-1.2);
   \node at (-0.7,-1.75) (t1) {$t_1$};
   \draw (0,-1.2) -- (-0.25,-2) -- (0.25,-2)-- (0,-1.2);
   \node at (0.05,-1.75) (t2) {$t_2$};
   \draw (0.75,-1.2) -- (0.5,-2) -- (1,-2)-- (0.75,-1.2);
   \node at (0.8,-1.75) (t3) {$t_3$};
   \draw (1.5,-1.2) -- (1.25,-2) -- (1.75,-2)-- (1.5,-1.2);
   \node at (1.55,-1.75) (t3) {$t_4$};
  \end{scope}
 \end{tikzpicture}
 \caption{\label{fuseadd}Operations \emph{fuse} and \emph{add}.}
\end{figure}

Note that the fuse and add operations preserve the embedding i.e., given a twig query $q$ and two trees $t$ and $t'$, if $t\preccurlyeq q$ and $t\unlhd t'$, then $t'\preccurlyeq q$.
Furthermore, if we can embed a query $q$ in a tree $t$ which can be embedded in the dependency graph of an MS $S$, we can perform a sequence of operations such that $t$ is transformed into another tree $t'$ satisfying $S$ and $q$ at the same time. Formally:
\begin{proposition}\label{L6}
Given an MS $S$, a query $q$ and a tree $t$, if $G_S\preccurlyeq t$ and $t\preccurlyeq q$, then there exists a tree $t'\in L (S)\cap L(q)$.
The tree $t'$ can be constructed after a sequence of fuse and add operations (consistently with the schema $S$) from the tree $t$ and we denote $t\unlhd_S t'$.
\end{proposition}

\subsubsection{Family of characteristic graphs}
Given a query $q$ and a schema $S$, if $q$ can be embedded in $G_S$ then we can capture all the trees satisfying $S$ and $q$ at the same time with a potentially infinite family of graphs.
First, we explain the construction of the characteristic graphs.
A \emph{characteristic graph $G$ for a schema $S$ and a query $q$} is a tuple $(V_G,\root_G,\lab_G,E_G)$, where $V_G$ is a finite set of vertices, $\root_G\in V_G$ is the root of the graph, $\lab_G:V_G\rightarrow\Sigma$ is a labeling function (with $\lab_G(\root_G) = \root_S$), and $E_G\subseteq V_G\times V_G$ represents the set of edges.
Note that for two $x,y\in\Sigma\cup\{\wc\}$ we say that $x$ \emph{matches} $y$ if $y\neq\wc$ implies $x=y$.
We construct $G$ with the three steps described below:
\begin{enumerate}
\item For any $(n_1,n_2)\in \child_q$, add $n_1', n_2'$ to $V_G$ and $(n_1', n_2')$ to $E_G$, where $\lab_G(n_1')$ matches $\lab_q(n_1)$ and $\lab_G(n_2')$ matches $\lab_q(n_2)$.
\item For any $(n_1, n_2)\in \desc_q$, choose an acyclic path $n_1',\ldots,n_k'$ from $G_S$, such that $n_1'$ matches $\lab_q(n_1)$ and $n_k'$ matches $\lab_q(n_2)$. 
We add to $G$ the corresponding vertices and edges for this path, as shown for the previous case.
\item For any $n\in V_G$, take the subgraph from $G_S^\mathrm{u}$ starting at $\lab_G(n)$ and fuse it in the node $n$ in the graph $G$.
\end{enumerate}
In Figure \ref{fig:graph} we present an example of graph obtained from the embedding from Figure \ref{fig:lambda}.
We denote by $\mathcal G(q,S)$ the set of all the graphs obtained from a query $q$ and a disjunction-free multiplicity schema $S$ using the three steps above, using all the embeddings from $q$ into $S$.
We extend the previous definition of the unfolding to the characteristic graphs.
Since a graph $G\in\mathcal G(q,S)$ is acyclic, it has a finite unfolding.
From the definition it also follows that the size of $G$ is polynomially bounded by $|q|\times |S|$ and $G\preccurlyeq q$.

If we allow cyclic paths in step 2, then we obtain similarly the set $\mathcal G^*(q,S)$.
Note that $|\mathcal G(q,S)|$ is finite and may be exponential, while $|\mathcal G^*(q,S)|$ may be infinite.
All the trees $t\in L(S)\cap  L(q)$ can be obtained by fuse and add operations (consistently with $S$) from the unfolding trees of the graphs in $\mathcal G^* (q,S)$:
\[
\forall t\in L(S)\cap L(q).\ \exists G\in \mathcal G^*(q,S).\ u_G \unlhd_S t
\]
Furthermore, by using a pumping argument, we have:
\[
\forall q\in\Twig.\ \forall G\in\mathcal G^*(q,S).\ (G\not\preccurlyeq q\Rightarrow  \exists G'\in\mathcal G(q,S).\ G'\not\preccurlyeq q).
\]

\begin{figure}[htb]
\centering
\subfigure[Embedding $\lambda:N_q\rightarrow G_S$.]{\label{fig:lambda}
\centering
\begin{tikzpicture}
\node at (-2,0) (s0) {$r$};
\node at (-2,-1) (s1) {$c$}edge[<-] (s0);
\node at (-3,-2) (s2) {$a_1$}edge[<-] (s1);
\node at (-1,-2) (s3) {$a_2$}edge[<-] (s1);
\node at (-2,-3) (s4) {$b$}edge[<-] (s2)edge[<-] (s3)edge[->,densely dotted] (s1);;
\node at (1,0) (n0) {$r$}edge[->,bend right, densely dotted] (s0);
\node at (0.5,-1) (n1) {$\wc$}edge[-,double] (n0)edge[->,bend right, densely dotted] (s3);
\node at (0.5,-2) (n2) {$\wc$}edge[-,double] (n1)edge[->,bend right, densely dotted] (s4);
\node at (0.5,-3) (n3) {$\wc$}edge[-,double] (n2)edge[->,bend left, densely dotted] (s3);
\node at (1.5,-1) (n4) {$b$}edge[-,double] (n0)edge[->,bend left, densely dotted] (s4);
\node at (1.5,-2) (n5) {$c$}edge[-] (n4)edge[->,bend right, densely dotted] (s1);
\end{tikzpicture}
}
\subfigure[Graph $G\in\mathcal G(q,S)$]{\label{fig:graph}~~~~~~~~~~~~~
\centering
\begin{tikzpicture}
\node at (0,0) (g0) {$\framebox{$r$}$};
\node at (1,0) (g00) {$c$} edge[<-] (g0);
\node at (1.9,0.25) (g01) {$a_1$} edge [<-] (g00);
\node at (1.9,-0.25) (g02) {$a_2$} edge [<-] (g00);
\node at (2.8,0) (g03) {$b$} edge [<-] (g01)edge [<-] (g02);
\node at (-2,-1) (g1) {$c$} edge[<-] (g0);
\node at (-1,-0.75) (g10) {$a_1$} edge [<-] (g1);
\node at (-1,-1.25) (g11) {$a_2$} edge [<-] (g1);\node at (-0.1,-1) (g12) {$b$} edge [<-] (g10)edge [<-] (g11);
\node at (-2,-2) (g2) {$\framebox{$a_2$}$} edge[<-] (g1);
\node at (-1.1,-2) (g20) {$b$} edge[<-] (g2);
\node at (-2,-3) (g3) {$\framebox{$b$}$} edge[<-] (g2);
\node at (-2,-4) (g4) {$c$} edge[<-] (g3);
\node at (-1,-3.75) (g40) {$a_1$} edge [<-] (g4);
\node at (-1,-4.25) (g41) {$a_2$} edge [<-] (g4);
\node at (-0.1,-4) (g42) {$b$} edge [<-] (g40)edge [<-] (g41);
\node at (-2,-5) (g5) {$\framebox{$a_2$}$} edge[<-] (g4);
\node at (-1,-5) (g50) {$b$} edge[<-] (g5);
\node at (1,-1) (g6) {$c$} edge[<-] (g0);
\node at (2,-0.75) (g60) {$a_1$} edge [<-] (g6);
\node at (2,-1.25) (g61) {$a_2$} edge [<-] (g6);
\node at (2.9,-1) (g62) {$b$} edge [<-] (g60)edge [<-] (g61);
\node at (1,-2) (g7) {$a_1$} edge[<-] (g6);
\node at (1.9,-2) (g70) {$b$} edge[<-] (g7);
\node at (1,-3) (g8) {$\framebox{$b$}$} edge[<-] (g7);
\node at (1,-4) (g9) {$\framebox{$c$}$} edge[<-] (g8);
\node at (2,-3.75) (g90) {$a_1$} edge [<-] (g9);
\node at (2,-4.25) (g91) {$a_2$} edge [<-] (g9);
\node at (2.9,-4) (g92) {$b$} edge [<-] (g90)edge [<-] (g91);
\end{tikzpicture}
}
\caption{\label{conp-alg}An embedding from a query $q$ to a dependency graph $G_S$ and a graph $G\in \mathcal G(q,S)$.
In $G_S$, the non-nullable edges are drawn with a full line and the nullable edges with a dotted line.}
\end{figure}
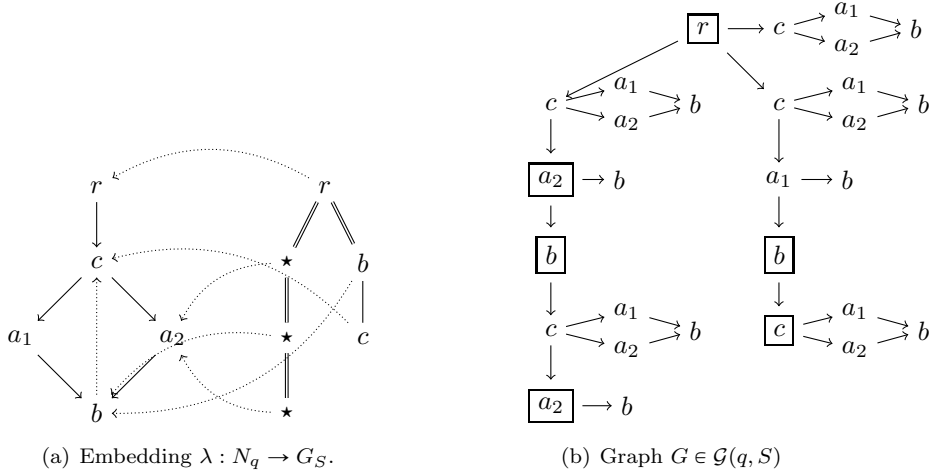
\subsubsection{Complexity results}
The dependency graphs and embeddings capture satisfiability and
implication of queries by MS.
\begin{lemma}\label{lemma:graph}
  For a twig query $q$ and an MS $S$ we have: 1) $q$ is satisfiable by
  $S$ iff $G_S\preccurlyeq q$, 2) $q$ is implied by $S$ iff
  $G_S^\mathrm{u}\preccurlyeq q$.
\end{lemma}
\begin{proof}\normalfont[sketch]
\noindent{\bf (1)} For the \emph{if} part, we know that $G_S\preccurlyeq q$, so the family of graphs $\mathcal G(q, S)$ is not empty.
The unfolding of any graph from $\mathcal G(q, S)$ satisfies $S$ and $q$ at the same time, hence $q$ is satisfiable by $S$.

For the \emph{only if} part, we know that there exists a tree $t\in L(S)\cap L(q)$, which  can for example be obtained after fuse operations (since one occurrence is consistent to all the multiplicities except 0) on the unfolding of a graph $G$ from $\mathcal G^*(q,S)$.
Since $t\preccurlyeq q$, we obtain $u_{G}\preccurlyeq q$, so $G\preccurlyeq q$, which, from the construction of $G$, implies that $G_S\preccurlyeq q$.

\noindent{\bf (2)} For the \emph{if} part, we know that $G_S^{\mathrm u}\preccurlyeq q$, which implies by Lemma~\ref{L3} that $u_{G_S^{\mathrm u}}\preccurlyeq q$.
On the other hand, take a tree $t\in L(S)$.
By Lemma~\ref{L1} we have $t\preccurlyeq G_{S}^{\mathrm u}$, which implies by Lemma~\ref{L2} that $t\preccurlyeq u_{G_S^{\mathrm u}}$.
From the last embedding and $u_{G_S^{\mathrm u}}\preccurlyeq q$ we infer that $t\preccurlyeq q$.
Since $t$ can be any tree in the language of $S$, we conclude that $q$ is implied by $S$.

For the \emph{only if} part, we know that for any $t\in L(S)$, $t\preccurlyeq q$.
Naturally, $u_{G_S^{\mathrm u}}$ is in the language of $S$ (since one occurrence is consistent to all
the multiplicities except 0), so $u_{G_S^{\mathrm u}}\preccurlyeq q$.
From the definition of the unfolding, we can infer that $G_S^{\mathrm u}\preccurlyeq u_{G_S^{\mathrm u}}$, which implies that $G_S^{\mathrm u}\preccurlyeq q$.\qed
\end{proof}
Furthermore, testing the embedding of a query in a graph can be done
in polynomial time with a simple bottom-up algorithm. From this observation and Lemma~\ref{lemma:graph}, we obtain:
\begin{theorem}\label{sat-impl-ptime}
$\SAT_{\ms,\Twig}$ and $\IMPL_{\ms,\Twig}$ are in PTIME.
\end{theorem}
The intractability of the containment of twig queries~\cite{MiSu04}
implies the coNP-hardness of the containment of twig queries in the
presence of MS.  Proving the membership of the problem to coNP is,
however, not trivial. Given an instance $(p,q,S)$, the set of all the
trees satisfying $p$ and $S$ can be characterized with a set
$\mathcal{G}({p,S})$ containing an exponential number of
polynomially-sized graphs and $p$ is contained in $q$ in the presence
of $S$ iff the query $q$ can be embedded into all the graphs in
$\mathcal{G}({p,S})$. This condition is easily checked by a
non-deterministic Turing machine.
\begin{theorem}\label{th:cnt-ms-hard}
$\CNT_{\ms,\Twig}$ is coNP-complete.
\end{theorem}
\begin{proof}\normalfont[sketch]
Theorem 4 from~\cite{MiSu04} implies that $\CNT_{\ms,\Twig}$ is coNP-hard.
Next, we prove the membership of the problem to coNP.
Given an instance $(p,q,S)$, a witness is a function $\lambda:N_p\rightarrow \Sigma$.
Testing whether $\lambda$ is an embedding from $p$ to $G_S$ requires polynomial time.
If $\lambda$ is an embedding, a non-deterministic polynomial algorithm chooses a graph $G$ from $\mathcal G(p,S)$ and checks whether $q$ can be embedded in $G$.
We claim that:
\[
p\varnot\subseteq_S q\iff \exists G\in\mathcal G(p,S).\ G\not\preccurlyeq q
\]
For the \emph{if} case, we assume that there exists a graph $G\in\mathcal G(p,S)$ such that $G\not\preccurlyeq q$.
We know that $G\preccurlyeq p$, so $u_G\preccurlyeq p$, so there exists a tree $t\in L(S)$ such that $t\preccurlyeq p$ and $u_G\unlhd_S t$ (using only fusions since one occurrence is consistent to all
the multiplicities except 0).
If we assume by absurd that $t\preccurlyeq q$, we have $u_G\preccurlyeq q$, so $G\preccurlyeq q$, which is a contradiction.
We infer thus that there exists a tree $t\in L(S)\cap  L(p)$, such that $t\notin L(q)$, so $p\varnot\subseteq_S q$.

For the \emph{only if} case, we assume that $p\varnot\subseteq_S q$, so there exists a tree $t\in L(S)\cap  L(p)$ such that $t\notin L(q)$.
Because $t\in L(S)\cap  L(p)$, we know that there exists a graph $G\in\mathcal G^* (p,S)$, such that $u_G\unlhd_S t$.
We know that $t\not\preccurlyeq q$, so $u_G\not\preccurlyeq q$, so $G\varnot\preccurlyeq q$.
Moreover, we know using the pumping argument that in this case there exists a graph $G'\in\mathcal G(p,S)$ such that $G'\not\preccurlyeq q$.
\qed
\end{proof}

\subsubsection{Extending the complexity results to disjunction-free DTDs}
We also point out that the complexity results for implication and containment of twig queries in the presence of MS can be adapted to disjunction-free DTDs.
This allows us to state results which, to the best of our knowledge, are novel.

Similarly to the MS, we represent a \emph{disjunction-free DTD} as a tuple $S=(\root_S, R_S)$, where $\root_S$ is a designed root label and $R_S$ maps symbols to regular expressions using no disjunction i.e., regular expressions of the form:
\[
E ::= \varepsilon\mid a\mid E^*\mid E^?\mid E^+\mid E_1\cdot E_2,
\]
where $a\in\Sigma$. Given such an expression $E$, consider the set $\universal(E)$ which contains the set of labels present in all the words from $ L(E)$. 
Formally,
\[
\universal(E) = \{a\in\Sigma\mid \forall w\in L(E).\ \exists w_1,w_2.\ w = w_1\cdot a\cdot w_2\}
\]
We can compute $\universal(E)$ recursively:
\begin{flalign*}
&\universal(\varepsilon) = \universal(E^*) = \universal(E^?) = \emptyset&\\
&\universal(a) = \{a\}\\
&\universal(E_1\cdot E_2) = \universal(E_1)\cup \universal(E_2)\\
&\universal(E^+) = \universal(E)
\end{flalign*}
Similarly, let $\existential(E)$ the set containing  labels which appear in at least one word from $ L(E)$. 
Formally,
\[
\existential(E) = \{a\in\Sigma\mid \exists w\in L(E).\ \exists w_1,w_2.\ w = w_1\cdot a\cdot w_2\}
\]
We can compute $\existential(E)$ recursively:
\begin{flalign*}
&\existential(\varepsilon) = \emptyset&\\
&\existential(a) = \{a\}\\
&\existential(E^{+/*/?})= \existential(E)\\
&\existential(E_1\cdot E_2) = \existential(E_1)\cup \existential(E_2)
\end{flalign*}
Next, we adapt the notions of dependency graph and universal dependency graph for disjunction-free DTDs.
The \emph{dependency graph} of a disjunction-free DTD $S$ is a rooted graph $G_S=(\Sigma, \root_S, E_S)$, where
\[
E_S = \{(a,a')\mid a'\in\existential(R_S(a))\}.
\]
Similarly, the \emph{universal dependency graph} of a disjunction-free DTD $S$ is a rooted graph $G_S^{\mathrm u}=(\Sigma, \root_S, E_S^{\mathrm u})$, where
\[
E_S^{\mathrm u} = \{(a,a')\mid a'\in\universal(R_S(a))\}.
\]
We assume w.l.o.g.\ that from now on we manipulate only disjunction-free DTDs having no cycle in the universal dependency graph.
Otherwise, if there is a cycle in the universal dependency graph, this means that there does not exist any tree consistent with the schema and containing any of the labels implied in that cycle.

For a symbol $a\in\Sigma$ and a disjunction-free regular expression $E$, by $\minnb(E,a)$ we denote the minimum number of occurrences of the symbol $a$ in any word consistent with $E$.
\begin{flalign*}
&\minnb(\varepsilon,a) = \minnb(E^*,a) = \minnb(E^?,a) = 0&\\
&\minnb(a,a) = 1\\
&\minnb(E_1\cdot E_2,a) = \minnb(E_1,a)+\minnb(E_2,a)\\
&\minnb(E^+,a) = \minnb(E,a)
\end{flalign*}
We adapt the definition of \emph{unfolding} for the (universal) dependency graph of a disjunction-free DTD.
For a disjunction-free multiplicity schema, the unfolding of the universal dependency graph belongs to its language since one occurrence is consistent with all the multiplicities except 0.
On the other hand, for a disjunction-free DTD $S$ this property does not hold, so we extend the construction of the unfolding with one more step:
\begin{itemize}
\item Let $u_{G_S^{\mathrm u}}$ be the unfolding of $G_S^{\mathrm u}$ obtained as it is defined for the MS.
\item Update $u_{G_S^{\mathrm u}}$ such that for any $n\in N_{u_{G_S^{\mathrm u}}}$, for any $a\in\Sigma$, 
let $t_a$ the subtree having as root the child of $n$ labeled by $a$.
Next, add copies of $t_a$ as children of $n$ until $n$ has $\minnb(R_S(\lab_{u_{G_S^{\mathrm u}}}(n)), a)$ children labeled with $a$.
\end{itemize}
Note that a consequence of this new definition is that the unfolding of the universal dependency graph of a disjunction-free DTD belongs to its language (modulo the order of the elements).
The order imposed by the DTD on the elements is not important because in the sequel we work with twig queries, which ignore this order.

\begin{corollary}\label{cor}
  $\IMPL_{\mathit{disj\text{-}free}\text{-}\DTD,\Twig}$ is in PTIME and
  $\CNT_{\mathit{disj\text{-}free}\text{-}\DTD,\Twig}$ is coNP-complete.
\end{corollary}
\begin{proof}\normalfont[sketch]
We claim that a query $q$ is implied by a disjunction-free DTD $S$ iff $G_S^{\mathrm u}\preccurlyeq q$ and since the embedding of a query in a graph can be computed in polynomial time, this implies that $\IMPL_{\mathit{disj\text{-}free}\text{-}\DTD,\Twig}$ is in PTIME.
The proof follows immediately from the proof of Lemma~\ref{lemma:graph}(2), taking into account the new definition of the unfolding.
Theorem 4 from~\cite{MiSu04} implies that $\CNT_{\mathit{disj\text{-}free}\text{-}\DTD,\Twig}$ is coNP-hard.
The membership of $\CNT_{\mathit{disj\text{-}free}\text{-}\DTD,\Twig}$ to coNP follows from the proof of Theorem~\ref{th:cnt-ms-hard}, while taking into account the new definition of the unfolding.
\qed
\end{proof}

\section{Expressiveness of DMS}
\label{sec:expressiveness}
\noindent We compare the expressive power of DMS and DTDs with focus
on schemas used in real-life applications. First, we introduce a simple
tool for comparing regular expressions with disjunctive multiplicity
expressions, and by extension, DTDs with DMS. For a regular expression
$R$, the language $L(R)$ of unordered words is obtained by removing
the relative order of symbols from every ordered word defined by
$R$. A disjunctive multiplicity expression $E$ \emph{captures} $R$ if
$L(E)=L(R)$. A DMS $S$ \emph{captures} a DTD $D$ if for every symbol
the disjunctive multiplicity expression on the rhs of a rule in $S$
captures the regular expression on the rhs of the corresponding rule
in $D$. We believe that this simple comparison is adequate because if
a DTD is to be used in a data-centric application, then supposedly the
order between siblings is not important. Therefore, a DMS that
captures a given DTD defines basically the same type of admissible
documents, without imposing any order among siblings. Naturally, if we
use the above notion to compare the expressive powers of DTDs and DMS,
DTDs are strictly more expressive than DMS.


We use the comparison on the XMark~\cite{SWKCMB02} benchmark and the
University of Amsterdam XML Web Collection~\cite{GrMa11}. We find that
all 77 regular expressions of the XMark benchmark are captured by DMS
rules, and among them 76 by MS rules. As for the DTDs found in the
University of Amsterdam XML Web Collection, $84\%$ of regular
expressions (with repetitions discarded) are captured by DMS rules and
among them $74.6\%$ by MS rules. Moreover, $55.5\%$ of full DTDs in
the collection are captured by DMS and among them $45.8\%$ by MS. Note
that these figures should be interpreted with caution, as we do not
know which of the considered DTDs were indeed intended for
data-centric applications. We believe, however, that these numbers
give a generally positive answer to the question of how much of the
expressive power of DTDs the proposed schema formalisms, DMS and MS,
retain.

\section{Conclusions and future work}\label{sec:conclusions}


We have studied the computational properties and the expressive power
of new schema formalisms, designed for unordered XML: the disjunctive
multiplicity schema (DMS) and its restriction, the disjunction-free
multiplicity schema (MS). DMS and MS can be seen as DTDs using
restricted classes of regular expressions and interpreted under
commutative closure to define unordered content models. These
restrictions allow on the one hand to maintain a relatively low
computational complexity of basic static analysis problems while
retaining a significant part of expressive power of DTDs. 

An interesting question remains open: are these the most general
restrictions that allow to maintain a low complexity profile? We
believe that the answer to this question is negative and intend to identify
new practical features that could be added to DMS and MS. One such
feature are \emph{numeric occurrences}~\cite{KiTu07} of the form
$a^{[n,m]}$ that generalize multiplicities by requiring the presence of
at least $n$ and no more than $m$ elements $a$. It would also be
interesting to see to what extent our results can be used to propose
hybrid schemas that allow to define ordered content for some elements
and unordered model for others.




 \bibliographystyle{plain}
 \bibliography{schema}

\begin{thebibliography}{10}

\bibitem{AbBoVi12}
S.~Abiteboul, P.~Bourhis, and V.~Vianu.
\newblock Highly expressive query languages for unordered data trees.
\newblock In {\em ICDT}, pages 46--60, 2012.

\bibitem{ACLS02}
S.~Amer-Yahia, S.~Cho, L.~V.~S. Lakshmanan, and D.~Srivastava.
\newblock Tree pattern query minimization.
\newblock {\em VLDB J.}, 11(4):315--331, 2002.

\bibitem{BeMi99}
C.~Beeri and T.~Milo.
\newblock Schemas for integration and translation of structured and
  semi-structured data.
\newblock In {\em ICDT}, pages 296--313, 1999.

\bibitem{BeFaGe08}
M.~Benedikt, W.~Fan, and F.~Geerts.
\newblock {XPath} satisfiability in the presence of {DTDs}.
\newblock {\em J. ACM}, 55(2), 2008.

\bibitem{BeBjHo11}
M.~Berglund, H.~Bj{\"o}rklund, and J.~H{\"o}gberg.
\newblock Recognizing shuffled languages.
\newblock In {\em LATA}, pages 142--154, 2011.

\bibitem{BeNeVa04}
G.~Bex, F.~Neven, and J.~Van~den Bussche.
\newblock {DTDs} versus {XML Schema}: A practical study.
\newblock In {\em WebDB}, pages 79--84, 2004.

\bibitem{BT05}
I.~Boneva and J.~Talbot.
\newblock Automata and logics for unranked and unordered trees.
\newblock In {\em RTA}, pages 500--515, 2005.

\bibitem{BTT05}
I.~Boneva, J.~Talbot, and S.~Tison.
\newblock Expressiveness of a spatial logic for trees.
\newblock In {\em LICS}, pages 280--289, 2005.

\bibitem{BuWo98}
A.~Br{\"u}ggemann-Klein and D.~Wood.
\newblock One-unambiguous regular languages.
\newblock {\em Inf. Comput.}, 142(2):182--206, 1998.

\bibitem{CG04}
L.~Cardelli and G.~Ghelli.
\newblock {TQL}: a query language for semistructured data based on the ambient
  logic.
\newblock {\em Mathematical Structures in Computer Science}, 14(3):285--327,
  2004.

\bibitem{DaLu03}
S.~Dal-Zilio and D.~Lugiez.
\newblock {XML} schema, tree logic and sheaves automata.
\newblock In {\em RTA}, pages 246--263, 2003.

\bibitem{GeMaNe09}
W.~Gelade, W.~Martens, and F.~Neven.
\newblock Optimizing schema languages for {XML}: Numerical constraints and
  interleaving.
\newblock {\em SIAM J. Comput.}, 38(5):2021--2043, 2009.

\bibitem{GrMa11}
S.~Grijzenhout and M.~Marx.
\newblock The quality of the {XML} web.
\newblock In {\em CIKM}, pages 1719--1724, 2011.

\bibitem{KiTu07}
P.~Kilpel{\"a}inen and R.~Tuhkanen.
\newblock One-unambiguity of regular expressions with numeric occurrence
  indicators.
\newblock {\em Inf. Comput.}, 205(6):890--916, 2007.

\bibitem{KoTo10}
E.~Kopczynski and A.~To.
\newblock Parikh images of grammars: Complexity and applications.
\newblock In {\em LICS}, pages 80--89, 2010.

\bibitem{MaNeSc09}
W.~Martens, F.~Neven, and T.~Schwentick.
\newblock Complexity of decision problems for {XML} schemas and chain regular
  expressions.
\newblock {\em SIAM J. Comput.}, 39(4):1486--1530, 2009.

\bibitem{MiSu04}
G.~Miklau and D.~Suciu.
\newblock Containment and equivalence for a fragment of {XPath}.
\newblock {\em J. ACM}, 51(1):2--45, 2004.

\bibitem{NeSc99}
F.~Neven and T.~Schwentick.
\newblock {XML} schemas without order.
\newblock 1999.

\bibitem{NeSc06}
F.~Neven and T.~Schwentick.
\newblock On the complexity of {XPath} containment in the presence of
  disjunction, {DTDs}, and variables.
\newblock {\em Logical Methods in Computer Science}, 2(3), 2006.

\bibitem{SWKCMB02}
A.~Schmidt, F.~Waas, M.~Kersten, M.~Carey, I.~Manolescu, and R.~Busse.
\newblock {XMark}: A benchmark for {XML} data management.
\newblock In {\em VLDB}, pages 974--985, 2002.

\bibitem{Schwentick04b}
T.~Schwentick.
\newblock Trees, automata and {XML}.
\newblock In {\em PODS}, page 222, 2004.

\bibitem{SeSi07}
L.~Segoufin and C.~Sirangelo.
\newblock Constant-memory validation of streaming {XML} documents against
  {DTDs}.
\newblock In {\em ICDT}, pages 299--313, 2007.

\bibitem{SeVi02}
L.~Segoufin and V.~Vianu.
\newblock Validating streaming {XML} documents.
\newblock In {\em PODS}, pages 53--64, 2002.

\bibitem{SeScMu03}
H.~Seidl, T.~Schwentick, and A.~Muscholl.
\newblock Numerical document queries.
\newblock In {\em PODS}, pages 155--166, 2003.

\bibitem{XPath1}
W3C.
\newblock {XML Path} language ({XPath}) 1.0, 1999.

\end{thebibliography}

\end{document}